\documentclass[fleqn,10pt]{wlscirep}
\usepackage[utf8]{inputenc}
\usepackage[T1]{fontenc}
\usepackage{amsthm}
\newtheorem{theorem}{Theorem}
\newtheorem{lemma}{Lemma}

\newtheorem{assumption}{Assumption}

\usepackage{physics}

\usepackage{color}
\usepackage{ulem}

\newcommand{\Romannumeral}[1]{\uppercase\expandafter{\romannumeral#1}}

\title{Quantum phase classification via partial tomography-based quantum hypothesis testing}

\author[1,*]{Akira Tanji}
\author[2]{Hiroshi Yano}
\author[1,2,+]{Naoki Yamamoto}
\affil[1]{Department of Applied Physics and Physico-Informatics, Keio University, Hiyoshi 3-14-1, Kohoku, Yokohama 223-8522, Japan}
\affil[2]{Keio Quantum Computing Center, Keio University, Hiyoshi 3-14-1, Kohoku, Yokohama 223-8522, Japan}

\affil[*]{tanjikeio@keio.jp}

\affil[+]{yamamoto@appi.keio.ac.jp}

\begin{abstract}
Quantum phase classification is a fundamental problem in quantum many-body physics, traditionally approached using order parameters or quantum machine learning techniques such as quantum convolutional neural networks (QCNNs). 
However, these methods often require extensive prior knowledge of the system or large numbers of quantum state copies for reliable classification. 
In this work, we propose a classification algorithm based on the quantum Neyman-Pearson test, which is theoretically optimal for distinguishing between two quantum states. 
While directly constructing the quantum Neyman-Pearson test for many-body systems via full state tomography is intractable due to the exponential growth of the Hilbert space, we introduce a partitioning strategy that applies hypothesis tests to subsystems rather than the entire state, effectively reducing the required number of quantum state copies while maintaining classification accuracy. 
We validate our approach through numerical simulations, demonstrating its advantages over conventional methods, including the order parameter-based classifier, the QCNN, and the recently developed classical machine learning algorithm enhanced with quantum data. 
Our results show that the proposed method achieves lower classification error probabilities with fewer required quantum state copies compared to all of these baselines, while also reducing the training cost relative to the QCNN and the classical machine learning algorithm enhanced with quantum data, and further decreasing the classical computational time in comparison with the latter.
We additionally demonstrate scalability of our method in numerical experiments up to systems with 81 qubits.
These findings highlight the potential of quantum hypothesis testing as a powerful tool for quantum phase classification, particularly in experimental settings where quantum measurements are combined with classical post-processing.  
\end{abstract}
\begin{document}

\flushbottom
\maketitle
\thispagestyle{empty}

% \noindent Please note: Abbreviations should be introduced at the first mention in the main text – no abbreviations lists. Suggested structure of main text (not enforced) is provided below.

\section*{Introduction}

Quantum phase classification~\cite{Sachdev2011,Vojta2003,Belitz2000,Shopova2003,Sondhi1997} is a fundamental task for understanding the behavior of quantum many-body systems that undergo phase transitions. 
Unlike classical systems, where thermal fluctuations drive phase transitions, quantum phase transitions may occur even at zero temperature and are driven by changes in external parameters such as pressure, magnetic field, or chemical composition. 
These transitions are marked by quantum fluctuations and are characterized by changes in the ground state properties of the system. 
The study of quantum phase transitions provides insight into the critical behavior and universality classes that describe how different systems behave near critical points.
A key tool for distinguishing different phases of matter is an order parameter~\cite{Sachdev2011,Vojta2003,Belitz2000,Shopova2003,Sondhi1997,Shi2015,Haegeman2012,Pollmann2012}, particularly in the context of symmetry-breaking. 
In many cases, phases can be classified based on whether the system exhibits a non-zero local order parameter, which signals spontaneous symmetry breaking.
Mathematically, an order parameter $o$ is defined as the expectation value of an observable $O$, given by $o = \lim_{L\rightarrow\infty} \langle O \rangle$, where $L \to \infty$ indicates the thermodynamic limit and the notation $\langle \cdot \rangle$ denotes the expectation value with respect to the ground state of the system in the case of quantum phases.
The value of $o$ characterizes the phase of the system, such that $o \neq 0$ if the system is ordered, and $o = 0$ if it is disordered.
For example, the magnetization in a ferromagnet serves as an order parameter that distinguishes between the magnetically ordered and disordered phases~\cite{Sachdev2011,Vojta2003,Belitz2000}.
However, this traditional framework of local order parameters falls short when it comes to identifying topological phases. 
Topological phases~\cite{Haegeman2012,Pollmann2012,Haldane2017,Qi2011} of matter do not break any local symmetry and therefore cannot be characterized by a local order parameter. 
Instead, these phases are defined by global properties that reflect non-local correlations in the system.
For example, topological phases exhibit topological invariants such as Chern numbers and $\hat{Z}$ invariants that remain robust under continuous deformations of the system~\cite{Haldane2017,Qi2011}.

Recently, quantum machine learning (QML) has attracted increasing attention due to its potential advantages in learning from quantum data, supported by theoretical and numerical studies including applications to quantum phase classification~\cite{Biamonte2017,Huang2022,Huang2022manybody,Huang2021,Liu2025}.
One representative approach is a quantum convolutional neural network (QCNN)~\cite{Cong2019,Monaco2023,Liu2023,Pesah2021,Herrmann2022}.
QCNNs are quantum circuit architectures inspired by classical convolutional neural networks~\cite{LeCun1998,Lecun2015}, designed to analyze quantum data by applying a series of convolutional and pooling layers. 
QCNNs have been successfully applied to recognize various quantum phases, including symmetry-protected topological (SPT) phases, demonstrating their capability to identify complex quantum states.
A notable feature of QCNNs is their resilience to the barren plateau problem, which ensures that they remain trainable with gradient-based methods even for large system sizes~\cite{Pesah2021}. 
Additionally, QCNNs can be trained with small sets of labeled data, making them practical for tasks where training data is limited~\cite{Caro2022}. 
QCNNs have been implemented successfully on current quantum hardware, showing their feasibility for near-term quantum applications~\cite{Herrmann2022}.
A notable subclass of QCNNs is Exact QCNNs~\cite{Cong2019,Sander2025,Lake2025}, which provide an analytical solution for classifying quantum phases without requiring training. 
Numerical experiments have shown that Exact QCNNs exhibit advantages in sample complexity compared to traditional order parameter-based approaches, making them more effective in scenarios where minimizing the number of required measurements is crucial~\cite{Cong2019}. 
However, Exact QCNNs have the downside that they are known to exist only for a limited range of quantum phases~\cite{Cong2019,Sander2025,Lake2025}.

Other QML approaches have been developed that combine quantum experiments with classical processing, in contrast to the approaches that apply classical processing to train quantum models, such as the above-mentioned QCNN or many other applications \cite{Peruzzo2014,Farhi2014,Xiao2023,Xu2025,Cao2025}. 
%in contrast to many approaches~\cite{Cong2019,Peruzzo2014,Farhi2014,Xiao2023,Xu2025} where classical processing serves as a support tool for training quantum models
This approach belongs to the class of Quantum Enhanced Classical Simulation (QESIM) studied in Ref.~\cite{Cerezo2025}, where initial data acquisition is performed via quantum experiments in no more than polynomial time, producing polynomial-size data that enhances the classical simulation algorithm. 
One such approach involves the use of an efficient classical representation, known as the classical shadow~\cite{Huang2020}, to construct a kernel function for quantum phase classification, among other applications~\cite{Huang2022manybody}.
Another approach is the low-weight Pauli propagation algorithm~\cite{Angrisani2025,Bermejo2024}, which leverages classical shadows to efficiently simulate quantum circuits on classical hardware within the range of low-weight Pauli operators. 
This method has been successfully applied to quantum phase classification using QCNNs~\cite{Bermejo2024}, referred to here as low-weight QCNNs. 
The study shows that QCNNs often operate within a classically simulable regime when restricted to low-body observables, raising questions about their quantum advantage. 
These methods demonstrate the potential of hybrid quantum-classical strategies for analyzing quantum phases while also emphasizing the need to identify problems where quantum computation offers a distinct advantage.

Following these lines of research, we propose a more simplified and efficient classification algorithm based on the quantum Neyman-Pearson test (also known as the Helstrom-Holevo test)~\cite{Helstrom1969,Hayashi2006,Nagaoka2007}, which is optimal for distinguishing between general two quantum states in the context of quantum hypothesis testing.
The classical Neyman-Pearson framework has been applied in machine learning for classification problems, where it provides a principled way to control error probabilities in classification~\cite{Tong2018, Scott2005}. 
Moreover, recent work has explored the use of quantum Neyman-Pearson test in the context of quantum machine learning and learning-theoretic analyses~\cite{Lloyd2020,Cheng2025,Banchi2021,Banchi2024,Blank2020}. 
The challenge in this framework is how to construct the corresponding quantum measurement from the training quantum states; for instance, a naive way using the full state tomography does not work, due to the exponential growth of the Hilbert space with system size~\cite{Banchi2024}. 
In this paper, we propose a quantum algorithm that exploits partial state tomography~\cite{Zhao2021,Bonet-Monroig2020} to efficiently construct the approximate quantum Neyman-Pearson test for quantum classification problems, particularly the quantum phase classification.
Our method can be viewed as a classical simulation algorithm supported by quantum experiments, in the sense that it constructs measurements on test data through classical processing of information obtained via partial tomography, and thus belongs to the QESIM framework. 
This method reduces the number of copies of quantum states required for classification while improving classification accuracy by effectively extracting relevant phase information from local measurements, making our algorithm both scalable and highly reliable for near-term quantum experiments.
We validate our approach through numerical simulations, demonstrating its advantages over conventional methods such as the order parameter, the QCNN, and the Exact QCNN. 
Our results show that it achieves lower classification error probabilities across various quantum phase classification tasks, which we attribute to the use of an approximately constructed quantum Neyman-Pearson test. 
Additionally, our method significantly improves training efficiency compared to the QCNN and the low-weight QCNN, as it does not rely on gradient-based variational learning but benefits from the approximate quantum Neyman-Pearson test. 
Notably, it achieves lower validation losses while requiring fewer training data copies, particularly compared to the QCNN. 
We also find that our method reduces the classical computational time complexity compared to the low-weight QCNN.
Furthermore, it maintains high classification accuracy across different system sizes, including cases with up to 81 qubits. 
These results highlight the potential of our method as a practical and efficient approach for quantum phase classification in experimental settings where quantum experiments are combined with classical processing.

\section*{Methods}
\label{sec:Methods}

Quantum hypothesis testing (QHT)~\cite{Helstrom1969,Hayashi2006,Nagaoka2007,Tej2018,Keski2021} is a fundamental tool in quantum information theory, designed to identify the true state of a quantum system from a set of possible states. 
Extending classical hypothesis testing principles to the quantum domain, QHT uses density operators to represent states and aims to develop optimal measurement strategies to minimize the probability of making incorrect decisions. 
This framework is essential for applications such as quantum communication and sensing.
In the most basic scenario, QHT involves two hypotheses where the task is to distinguish between two potential states, $\rho$ (the null hypothesis) and $\sigma$ (the alternative hypothesis). 
A key challenge in this process is balancing and minimizing the types of errors that can occur. 
Specifically, the Type-\Romannumeral{1} error happens when $\sigma$ is chosen when $\rho$ is true, while the Type-\Romannumeral{2} error occurs when $\rho$ is selected when $\sigma$ is correct. 
In other words, the Type-\Romannumeral{1} error happens when the alternative hypothesis $\sigma$ is accepted despite the null hypothesis $\rho$ being the true state, whereas the Type-\Romannumeral{2} error occurs when the null hypothesis $\rho$ is chosen even though the alternative hypothesis $\sigma$ is the correct state. 
Mathematically, the Type-\Romannumeral{1} error probability, denoted as $\alpha_n$, is given by
\begin{equation}\label{eq:type_one_error}
    \alpha_n = \Tr(\rho^{\otimes n} (I-M_n)) = 1 - \Tr(\rho^{\otimes n} M_n),
\end{equation}
where $\{M_n, I-M_n\}$ are two-outcome Positive Operator-Valued Measures (POVMs) associated with deciding in favor of $\rho$ and $\sigma$, respectively. 
Similarly, the Type-\Romannumeral{2} error probability, denoted as $\beta_n$, is defined as
\begin{equation}\label{eq:type_two_error}
    \beta_n = \Tr(\sigma^{\otimes n} M_n).
\end{equation}
Note that, in the above definitions, quantum states are assumed to be prepared in multiple $n$ copies, denoted by $\rho^{\otimes n}$ and $\sigma^{\otimes n}$, which are collectively or individually measured to distinguish between the two hypotheses. 
Measurement strategies must be carefully designed to manage these errors; 
in particular, symmetric approaches aim to minimize the total error probability, and asymmetric strategies focus on reducing one type of error while controlling the other within acceptable bounds. 
The quantum Neyman-Pearson test $\{ S(n,a), I-S(n,a) \}$ is defined by the POVM element
\begin{equation}\label{eq:qNeyman}
S(n,a) = \qty{\rho^{\otimes n} - e^{na}\sigma^{\otimes n} > 0},
\end{equation}
where $a$ is a hyperparameter that controls the balance between Type-\Romannumeral{1} and Type-\Romannumeral{2} error probabilities. 
The notation $\{A>0\}$ represents the projection operator onto the subspace spanned by the eigenvectors of a Hermitian operator $A$ with positive eigenvalues, i.e.,
\begin{equation}
\{A>0\}=\sum_{i:\lambda_i>0}\ketbra{\lambda_i}, \quad \text{if } A = \sum_i \lambda_i\ketbra{\lambda_i}.
\end{equation}
This test is well known in the theory of quantum hypothesis testing as the most powerful test for distinguishing two quantum states. 
That is, the test maximizes the power $1-\beta_n$, which is the probability of correctly rejecting the null hypothesis $\rho$ when the alternative hypothesis $\sigma$ is true, while controlling the Type-\Romannumeral{1} error probability $\alpha_n$ to be below a given significance level. 
The same holds if the roles of $\alpha_n$ and $\beta_n$ are reversed. 
Moreover, this test is employed to achieve the asymptotic optimality described by the quantum Stein’s lemma~\cite{Hiai1991,Ogawa2000} and the quantum Hoeffding’s theorem~\cite{Ogawa2004}.
Note that, when the two quantum states $\rho$ and $\sigma$ commute, the quantum Neyman-Pearson test is equivalent to the classical Neyman-Pearson test as described in Appendix. 
However, constructing this measurement with full state tomography requires an exponential number of copies of the quantum states or classical processing with respect to system size, making it impractical for large-scale quantum systems.

% We introduce a POVM element $M_n = S_{n_{\text{ent}}}^{\otimes (n/n_{\text{ent}})}(a)$ defined by
% \begin{equation}\label{eq:qNeyman}
% S_{n_{\text{ent}}}(a) = \qty{\rho^{\otimes n_{\text{ent}}} - e^{n_{\text{ent}}a}\sigma^{\otimes n_{\text{ent}}} > 0},
% \end{equation}
% where $n_{\text{ent}}$ represents the number of copies of the quantum states $\rho$ and $\sigma$, and $a$ is a hyperparameter that controls the balance between Type-\Romannumeral{1} and Type-\Romannumeral{2} error probabilities. 
% For simplicity, we assume throughout this paper that $n$ is divisible by $n_\mathrm{ent}$.
% The notation $\{A>0\}$ represents the projection operator onto the subspace spanned by the eigenvectors of a Hermitian operator $A$ with positive eigenvalues, i.e.,
% \begin{equation}
% \{A>0\}=\sum_{i:\lambda_i>0}\ketbra{\lambda_i}, \quad \text{if } A = \sum_i \lambda_i\ketbra{\lambda_i}.
% \end{equation}
% When $n_{\text{ent}} = n$ and $\rho$ and $\sigma$ are the quantum states being tested, this corresponds to the quantum Neyman-Pearson test \red{\sout{(or Helstrom measurement)}}, which is well known in quantum hypothesis testing as the most powerful test for distinguishing two quantum states. 
% However, constructing this measurement with full state tomography requires an exponential number of copies of the quantum states or classical processing with respect to system size, making it impractical for large-scale quantum systems.

\begin{figure}[t]
    \centering
    \includegraphics[width=0.62\columnwidth]{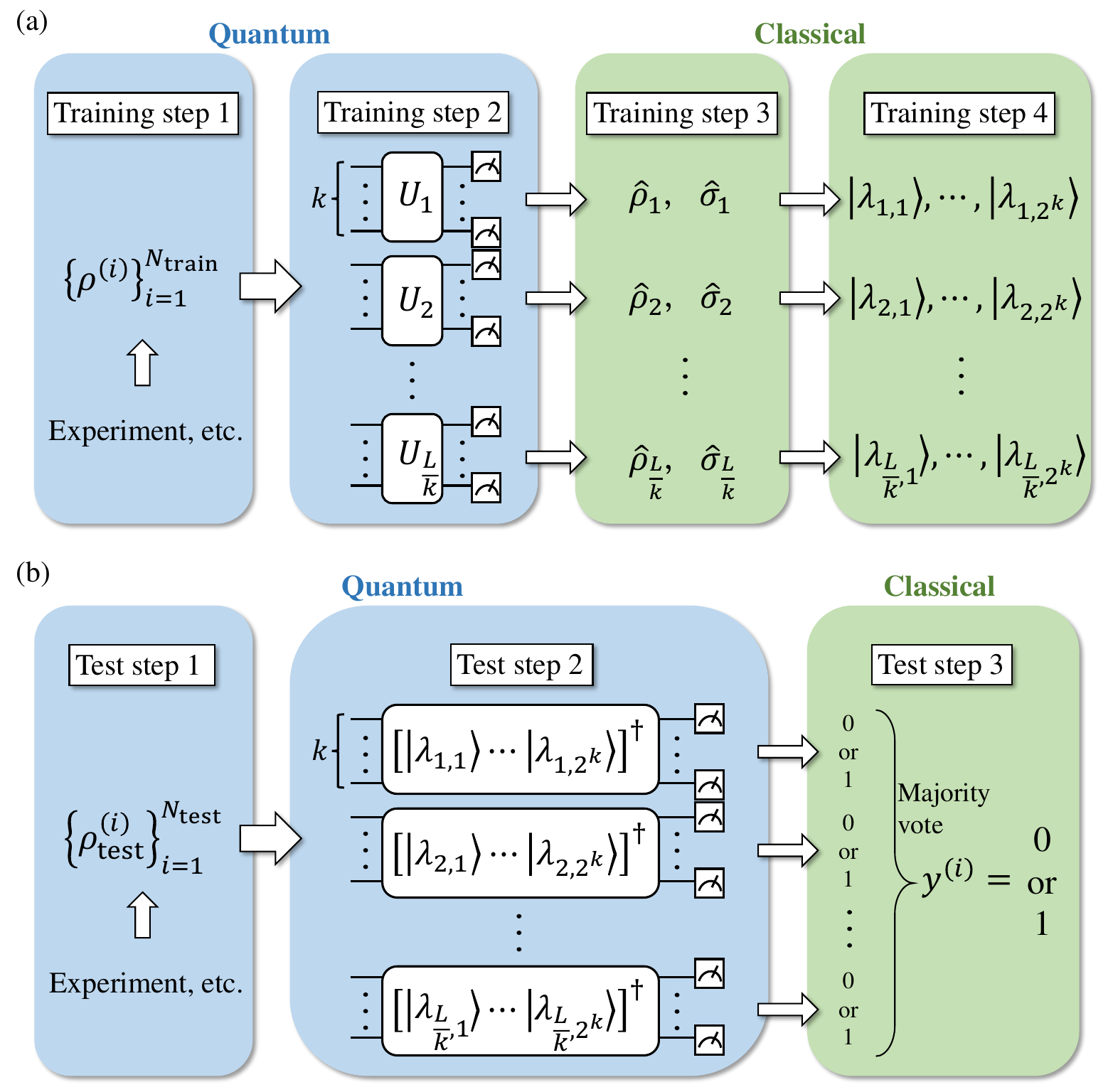}
    \caption{Schematic depiction of our method in the (a) training step and (b) test step. 
    Steps within the blue boxes represent quantum processes, while those within the green boxes represent classical processes. 
    In the training step (a), we perform partial tomography on all the quantum many-body states $\{\rho^{(i)}\}$ in a training dataset, obtaining their $k$-RDMs. 
    In the test step (b), for each test state $\rho_{\text{test}}^{(i)}$, the approximate quantum Neyman-Pearson test is conducted to obtain a prediction $y^{(i)}$.}
    \label{fig:Method}
\end{figure}

To overcome this scalability issue, we propose a method for approximately constructing a quantum Neyman-Pearson test via the decomposition of a quantum many-body state into smaller subsystems followed by their partial state tomography. 
Applying this method to the quantum phase classification problem allows us to reduce the error probability for single-copy data near the phase boundary and the number of copies of data needed to construct the measurement.
The key assumption is that target quantum phases could potentially be classified by dividing quantum many-body systems into groups of a few qubits. 
This stems from the fact that most condensed matter Hamiltonians exhibit few-body interactions, meaning they involve interactions that only extend among a few particles. 
Furthermore, many order parameters can be represented as linear combinations of local observables~\cite{Sachdev2011,Vojta2003,Shi2015}. 
Even in topological phases, which are often characterized by global order parameters~\cite{Haldane2017,Qi2011}, classification can sometimes be achieved using few-body observables~\cite{Huang2022manybody}.
This does not imply the existence of a local order parameter for topological phases, but only indicates that, in certain models, phases are identifiable from few-body information.

We here describe our classification method; the entire procedure is depicted in Fig.~\ref{fig:Method}. 
The key concept is a partitioned quantum Neyman-Pearson test with POVM
\begin{equation}\label{eq:partitioned qNeyman}
\begin{split}
S_j^{(0)}(n_{\text{ent}}, a) &= \qty{\rho_j^{\otimes n_{\text{ent}}} - e^{n_{\text{ent}}a} \sigma_j^{\otimes n_{\text{ent}}} > 0}, \\
S_j^{(1)}(n_{\text{ent}}, a) &= I - S_j^{(0)}(n_{\text{ent}}, a),
\end{split}
\end{equation}
where $\rho_j$ and $\sigma_j$ denote the $k$-qubit reduced density matrices ($k$-RDMs) of the full quantum states $\rho$ and $\sigma$ corresponding to the $j$-th $k$-qubit subsystem.
In the main text, we describe the method for the case $n_{\text{ent}}=1$; for the case $n_{\text{ent}}>1$, see the Additional numerical simulations section of the Appendix. 
The classifier, or equivalently the POVM corresponding to Eq.~\eqref{eq:partitioned qNeyman}, is constructed using given training data (a set of quantum states). 
We then apply the constructed POVM to classify new data, which we call test data. 
More specifically, the training step proceeds as follows, as shown in Fig.~\ref{fig:Method}(a):
\begin{enumerate}
    \item Given $N_{\text{train}}$ quantum many-body states $\{ \rho^{(i)}\}_{i=1}^{N_{\text{train}}}$ on $L$ qubits as training data and corresponding binary labels $\{ y^{(i)} \in \{0,1\} \}_{i=1}^{N_{\text{train}}}$.
    \item For each training data $\rho^{(i)}$, consider the $k$-RDMs $\{ \rho_j^{(i)} \}_{j=1}^{L/k}$ and estimate them as $\{ \hat{\rho}_j^{(i)} \}_{j=1}^{L/k}$ using partial tomography. (Here, we simply assumed that $L$ is divisible by $k$.)
    \item For each group $j$, calculate the ensemble average of matrices obtained via tomography for states labeled as 0, $\hat{\rho}_j = \sum_{i: y^{(i)}=0} p^{(i)} \hat{\rho}_j^{(i)}$, and for states labeled as 1, $\hat{\sigma}_j = \sum_{i: y^{(i)}=1} p^{(i)} \hat{\rho}_j^{(i)}$. 
    \item For each group $j$, compute the eigenvalues $\{ \lambda_{j,l} \}_{l=1}^{2^k}$ and eigenvectors $\{ \ket{\lambda_{j,l}} \}_{l=1}^{2^k}$ of $\hat{\rho}_j - e^a \hat{\sigma}_j$, and construct the gate implementation of the unitary matrix $[\ket{\lambda_{j,1}}, \ldots, \ket*{\lambda_{j,2^k}}]^\dagger$ to perform the partitioned quantum Neyman-Pearson test with POVMs $S_j^{(0)}(n_{\text{ent}}=1, a)= \{ \hat{\rho}_j - e^a \hat{\sigma}_j > 0 \} = \sum_{l: \lambda_{j,l} > 0} \ketbra{\lambda_{j,l}}$ and $S_j^{(1)}(n_{\text{ent}}=1, a)= \sum_{l: \lambda_{j,l} \leq 0} \ketbra{\lambda_{j,l}}$, which are created based on the estimators $\hat{\rho_j}$ and $\hat{\sigma_j}$. 
\end{enumerate}

The test step proceeds as shown in Fig.~\ref{fig:Method}(b):
\begin{enumerate}
    \item Given $N_{\text{test}}$ quantum many-body states $\{ \rho_{\text{test}}^{(i)} \}_{i=1}^{N_{\text{test}}}$ on $L$ qubits as test data.
    \item For each test data $\rho_{\text{test}}^{(i)}$, consider the $k$-RDMs $\{ \rho_{\text{test},j}^{(i)} \}_{j=1}^{L/k}$, and for each group $j$, perform the POVMs $\{ S_j^{(0)}(n_{\text{ent}}=1, a),~S_j^{(1)}(n_{\text{ent}}=1, a) \}$ corresponding to the quantum Neyman-Pearson test on each group, utilizing the unitary matrix obtained in Training step 4. 
    \item For each $i$, take a majority vote of the measurement results across all groups $j=1,...,L/k$, classifying the test data as $y^{(i)} = 0$ if the former POVM is more frequently measured, and as $y^{(i)} = 1$ otherwise.
\end{enumerate}
In Training step 1, we assume that quantum many-body states serving as training data can be prepared on a quantum device using quantum state preparation algorithms or experimental methods, with known corresponding phase labels. 
%We refer to quantum state tomography on each group in training step 2 as partial tomography.
% We obtain estimates of $k$-RDMs $\{ \hat{\rho}_j^{(i)} \}_{j=1}^{L/k}$ as density matrices or shadow estimators~\cite{Huang2020}; 
% \red{for the latter case, we use}\brown{[In both cases, if there is no prior knowledge, would we use a tomographically complete measurement?]} any tomographically complete measurement, particularly if there is no prior knowledge of the training data.
In Training step 2, we obtain estimates of $k$-RDMs $\{ \hat{\rho}_j^{(i)} \}_{j=1}^{L/k}$ via partial tomography techniques, such as a classical shadow~\cite{Huang2020}.
When no prior knowledge about the system is available, both the subsystem size $k$ and the way of partitioning can be treated as hyperparameters. These can be tuned using training data, for instance by optimizing the classification performance on validation datasets.
Training step 3 involves matrix calculations on a classical computer, where ensemble averaging corresponds to replacing composite quantum hypothesis tests with simple ones, with probabilities $p^{(i)}$ typically uniform in phase classification tasks but potentially variable.
Training step 4 also involves classical matrix calculations, constructing unitaries for performing the approximate quantum Neyman-Pearson test, %\brown{\sout{$\otimes_{j=1}^{L/k} S_{j, n_{\text{ent}}=1}(a) = \otimes_{j=1}^{L/k} \{ \hat{\rho}_j - e^a \hat{\sigma}_j > 0 \}$}}, 
the gate implementation of which can be efficiently constructed for small $k$-qubit groups~\cite{Dawson2006,Ge2024}.
Specifically, we perform an eigendecomposition of $\hat{\rho}_j - e^a \hat{\sigma}_j$ classically and construct the corresponding unitary matrix. 
This unitary is then naively embedded into the quantum circuit as if it were an amplitude encoding technique. 
As a result, the circuit depth is exponential in $k$ but independent of the whole system size $L$, and the classical computational time scales exponentially in $k$ and linearly in $L$.
Here, $a$ is an arbitrary real number, which is set to $0$ when no preference is given to either quantum phase. 
When the task requires emphasizing the detection of label 0 (1), such as in anomaly detection, one can set $a$ to a smaller (larger) value accordingly. 
To achieve a specific balance between Type-\Romannumeral{1} and Type-\Romannumeral{2} error probabilities, $a$ may be tuned based on training data, and the number of additional quantum state copies required for tuning can be reduced by performing a pre-tuning of $a$ based on estimated RDMs.
Although it is possible to set $n_{\text{ent}} > 1$, numerical results presented in the Additional numerical simulations section of the Appendix indicate that this had limited significance.

In Test step 1, we assume that test data can be prepared on a quantum device similarly to training data. 
Test step 2 implements the two-outcome POVMs for each group, as shown in Fig.~\ref{fig:Method}(b), by measuring in the basis $\{ \ket{\lambda_{j,l}} \}_{l=1}^{2^k}$ and post-selecting based on whether the eigenvalue $\lambda_{j,l}$ of the measured basis state is positive or not. 
Test step 3 completes the classification of the quantum phase by taking the majority vote POVM across all groups, which aggregates the results of partitioned quantum Neyman-Pearson tests applied to $k$-RDMs. 
It is defined as
\begin{equation}
M_n = \Pi_{\text{maj}}^{\otimes (n/n_{\text{ent}})}(n_{\text{ent}}, a),
\end{equation}
where $n_{\text{ent}}$ is the same number in Eq.~\eqref{eq:partitioned qNeyman}, i.e., the number of copies used per subsystem in each individual Neyman-Pearson test. 
The majority vote POVM element is given by
\begin{equation}
\Pi_{\text{maj}}(n_{\text{ent}}, a) = \sum_{\substack{x_1, \dots, x_{L/k} \\ \sum_j x_j < L/2k}} \bigotimes_{j=1}^{L/k} S_j^{(x_j)}(n_{\text{ent}}, a),
\end{equation}
where $L$ is the entire system size. 
In this step, each $j$-th $k$-qubit subsystem is tested independently using a partitioned quantum Neyman-Pearson test. 
Then, the majority vote POVM element $\Pi_{\text{maj}}(n_{\text{ent}}, a)$ selects the outcome for which the majority of the subsystems yield $S_j^{(0)}(n_{\text{ent}}, a)$. 
The introduction of the majority vote is motivated by the absence of a clear procedure for aggregating measurement results obtained from partitioned quantum Neyman-Pearson tests.
Since these tests are applied independently to each $k$-qubit subsystem, a majority vote offers a simple and natural way to combine local information into a global decision. 
Moreover, in the majority vote section of the Appendix, we provide a brief rationale for this choice by showing that, under certain conditions, the majority vote suppresses the variance of the estimator.
Although this does not yet constitute a fully rigorous justification, it supports the use of the majority vote as a reasonable and practically effective aggregation strategy in our framework.

\section*{Numerical results}
\label{sec:Results}
This section presents numerical results for quantum phase classification using various methods from the three aspects: error probabilities, training costs, and scalability. 
More precisely, we first evaluate the Type-\Romannumeral{1} and Type-\Romannumeral{2} error probabilities to assess classification performance. 
We then compare our method with the QCNN and the low-weight QCNN~\cite{Bermejo2024,Angrisani2025}, in terms of the training costs. 
Furthermore, we examine the scalability of our method by analyzing its performance across different system sizes up to 81 qubits.
In the Appendix, the settings of simulations and additional numerical simulations are provided. 

\subsection*{Model and settings}

We conduct numerical simulations for quantum phase classification of the one-dimensional cluster-Ising model
\begin{equation}\label{eq:cluster-Ising}
H = \sum^{L}_{i=1}\left(X_i - J_1Z_iZ_{i+1} - J_2Z_{i-1}X_iZ_{i+1}\right),
\end{equation}
where $X_i(Z_i)$ are the Pauli $X(Z)$ operators on the $i$-th qubit. 
$J_1$ and $J_2$ are tunable coupling coefficients. 
The ground states of this many-body Hamiltonian, as shown in Fig.~\ref{fig:phase cluster-Ising}, exhibit four distinct phases: ferromagnetic (FM), antiferromagnetic (AFM), symmetry-protected topological (SPT), and trivial~\cite{Verresen2017}. 
The order parameter for the FM phase 
\begin{equation}\label{eq:order FM}
O_{\text{FM}} = \frac{1}{L}\sum^{L}_{i=1}Z_i
\end{equation}
is a linear combination of local observables, and the trivial phase is similarly detected by a linear combination of local observables. 
In contrast, the order parameter for the SPT phase is expressed by a global observable 
\begin{equation}\label{eq:order SPT}
O_{\text{SPT}} = Z_1X_2X_4...X_{L-3}X_{L-1}Z_L,
\end{equation}
which characterizes the topology of the system~\cite{Haegeman2012,Pollmann2012}.
We refer to the classification between the trivial and FM phases as \textit{Trivial vs. FM}, and the classification between the trivial and SPT phases as \textit{Trivial vs. SPT}.
For each case, we employed five methods for quantum phase classification: order parameter, QCNN, Exact QCNN, low-weight QCNN, and our method. 
The test data, which are quantum states we want to phase classify, are common across all methods and consist of 100 ground states near the phase boundaries depicted in each of Fig.~\ref{fig:phase cluster-Ising}(a) and Fig.~\ref{fig:phase cluster-Ising}(b), respectively. 
The training data, needed to construct the QCNN and our method, comprise 20 ground states shown in Fig.~\ref{fig:phase cluster-Ising}(a) and Fig.~\ref{fig:phase cluster-Ising}(b), with labels assigned as $y^{(i)} = 0$ for the trivial phase and $y^{(i)} = 1$ for the other phase.
For numerical simulations of large system sizes ($L=27$ and 81 qubits), we utilize Matrix Product States (MPS)~\cite{Schollwöck2005, Vidal2007}. 
We employ the finite-size Density Matrix Renormalization Group (DMRG) algorithm~\cite{Schollwöck2005, Vidal2007} with the maximum bond dimension 200 to prepare approximate ground states in the MPS representation. 
On the other hand, for small system sizes ($L=15$ qubits), we use state vectors and the exact diagonalization algorithm.

Our method uses $k=2$ qubits per group in the case of Trivial vs. FM, and $k=3$ qubits in the case of Trivial vs. SPT. 
This choice is due to the fact that the FM phase and the SPT phase in the model described above can be classified using low-body observables~\cite{Huang2022manybody}.
For dividing the quantum many-body state, we assign $k$ neighboring qubits to the same group, starting from one end of the chain, since the model is one-dimensional. 
While this grouping approach is straightforward for one-dimensional models, it is not for models in two or higher dimensions. 
However, we show in the Additional numerical simulations section of the Appendix that even a simple dividing strategy works sufficiently well for two-dimensional models.
We employ partial tomography as Training step 2, utilizing a partitioned classical shadow~\cite{Huang2020} approach for its rapid convergence facilitated by random measurements. 
More precisely, for each partitioned group $j$, unitary operators $\{U_j\}$ are sampled uniformly at random from the $k$-qubit Clifford group, followed by computational basis measurements yielding outcome $\ket*{\hat{b}}$, to obtain snapshots $\{ (2^k + 1) U_j^\dagger \ketbra*{\hat{b}} U_j - I\}$. 
The sample average of these snapshots directly forms $\hat{\rho}_j^{(i)}$, constructing it efficiently from these measurements.

\begin{figure}[htp!]
    \centering
    \includegraphics[width=0.5\columnwidth]{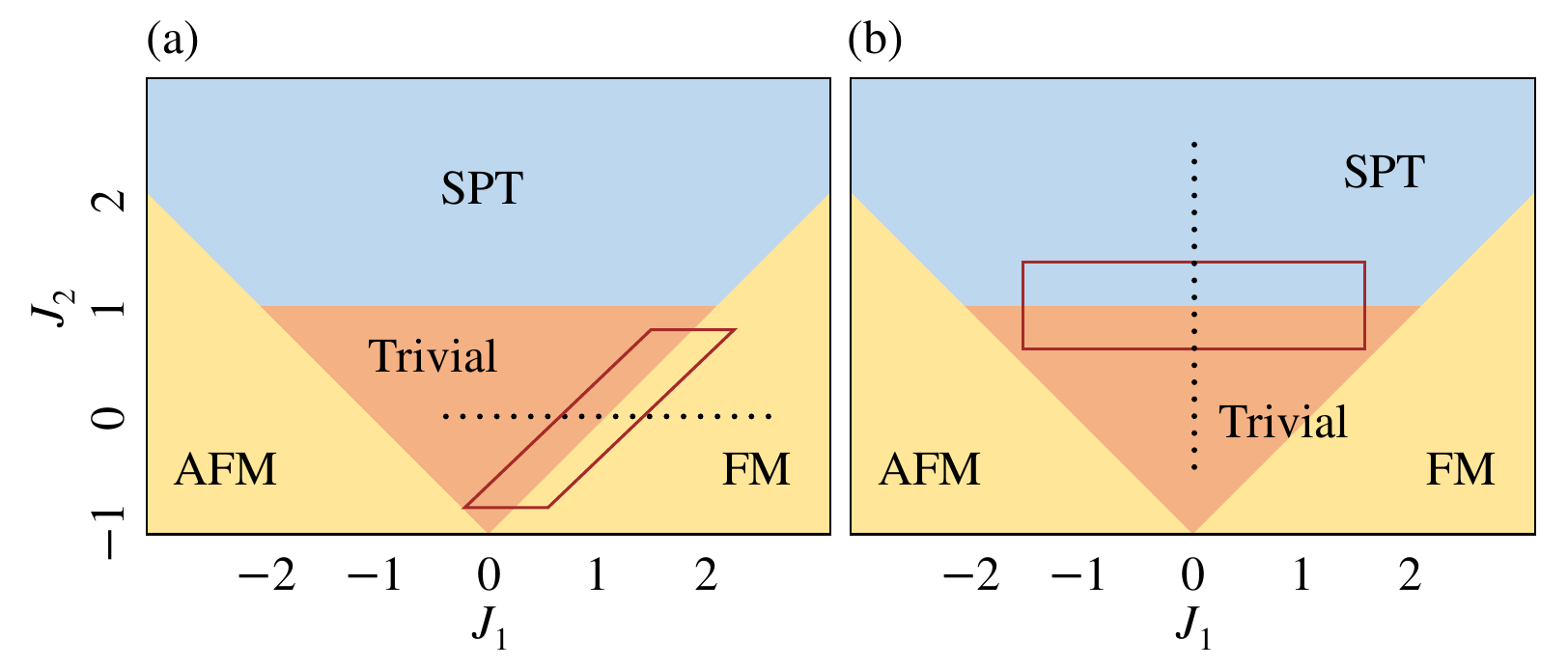}
    \caption{Quantum phase diagram of the ground state of the Hamiltonian in Eq.~
    (\ref{eq:cluster-Ising}), along with the training and test data for (a) Trivial vs. FM and (b) Trivial vs. SPT cases. 
    The 20 black dots represent the training data, while the test data consist of 100 randomly selected points within the red box near the phase boundary.}
    \label{fig:phase cluster-Ising}
\end{figure}

\subsection*{Error probabilities}

We calculate the Type-\Romannumeral{1} and Type-\Romannumeral{2} error probabilities, $\alpha_n$ defined in Eq. \eqref{eq:type_one_error} and $\beta_n$ in Eq. \eqref{eq:type_two_error}, for the test data of quantum phase classification in the cases of Trivial vs. FM and Trivial vs. SPT, using four methods: order parameter, QCNN, Exact QCNN, and our method. 
To calculate the error probabilities, we take the average error probabilities for all test data in each phase, meaning that we replace the composite hypotheses with simple ones.
The results for a 27-qubit system ($L = 27$) are shown in Fig.~\ref{fig:Error probabilities}. 
Since lower error probabilities with fewer copies indicate better performance, methods represented in the lower left of panel (a), and lower positions in panel (b), are preferable.
Panel (a) displays the trade-off between these two types of error probabilities for each methods. 
The balance between these error probabilities is controlled by a hyperparameter, which was varied over a wide range to produce the curves shown in the figure.
For our method, the hyperparameter $a$ in Eq.~\eqref{eq:partitioned qNeyman} and Training step 4 is selected from 20 evenly spaced values in the interval $[-1, 1]$. 
For the other methods, the relevant hyperparameters are finely discretized over a dense grid, as described in Appendix, resulting in curves rather than discrete points.
The total number of training data copies used for our method is 600 or 2,400; 
more precisely, for the former case, $\text{training shots} = N_{\text{train}} \times T_{\text{state}} = 20 \times 30 = 600$, where $T_{\text{state}}$ is the number of shots per state (or equivalently the number of snapshot for constructing the shadow) in the training dataset. 
In contrast, the QCNN requires exact calculations of output expectation values to ignore the estimation errors, resulting in a total of $N_{\text{train}} \times N_{\text{epoch}} \times 2 = 20 \times 150 \times 2 = 6{,}000$ evaluations, where $N_{\text{epoch}}$ is the number of epochs used for training the QCNN and $\times 2$ means two times expectation estimation involved in the optimizer described below. 
The training cost for our method is therefore significantly lower compared to that of the QCNN.
Additionally, the Simultaneous Perturbation Stochastic Approximation (SPSA) optimizer~\cite{Spall1998,Spall1998overview} is used for training the QCNN, where we estimate two expectation values to evaluate the gradient of the loss function. 
Since the number of training shots used for each expectation estimation is not explicitly considered here, the expectation values for the QCNN are computed exactly, which corresponds to taking the number of copies of training data per expectation value estimate as $T_{\text{est}} = \infty.$

Note that the problem settings differ between those for the order parameter and the Exact QCNN, and those for our method and the QCNN. 
The order parameter and the Exact QCNN operate without training data but with prior knowledge about the quantum phases, whereas our method and the QCNN rely on training data without any prior information on the phase. 
In the former setting, for the Trivial vs. SPT case, the Exact QCNN demonstrates significant improvements in error probabilities compared to the order parameter. 
Yet, for the Trivial vs. FM case, the order parameter combined with a Bayesian test (described in the Details of methods section of the Appendix) using an appropriately chosen prior distribution achieves lower error probabilities than the Exact QCNN. 
In the latter setting, in both cases, our method demonstrates substantial improvements over the QCNN in terms of error probabilities, despite using significantly fewer training data copies (as well as training costs, detailed in the next subsection).
While a direct comparison must be made with caution due to the above differences in problem settings, our method shows great potential for outperforming the order parameter and the Exact QCNN in terms of error probabilities for both cases.
Particularly for the Trivial vs. SPT case, it is especially remarkable that our method is expected to achieve much lower error probabilities than the commonly supported method, the order parameter.
These results are obtained for a small number of test data copies, specifically $n = 1$ or at most 20.
However, as suggested by Fig.~\ref{fig:Error probabilities}(b), our method is anticipated to achieve better performance than the other methods even for larger numbers of test data copies.

\begin{figure*}[htp!]
    \centering
    \includegraphics[width=1\columnwidth]{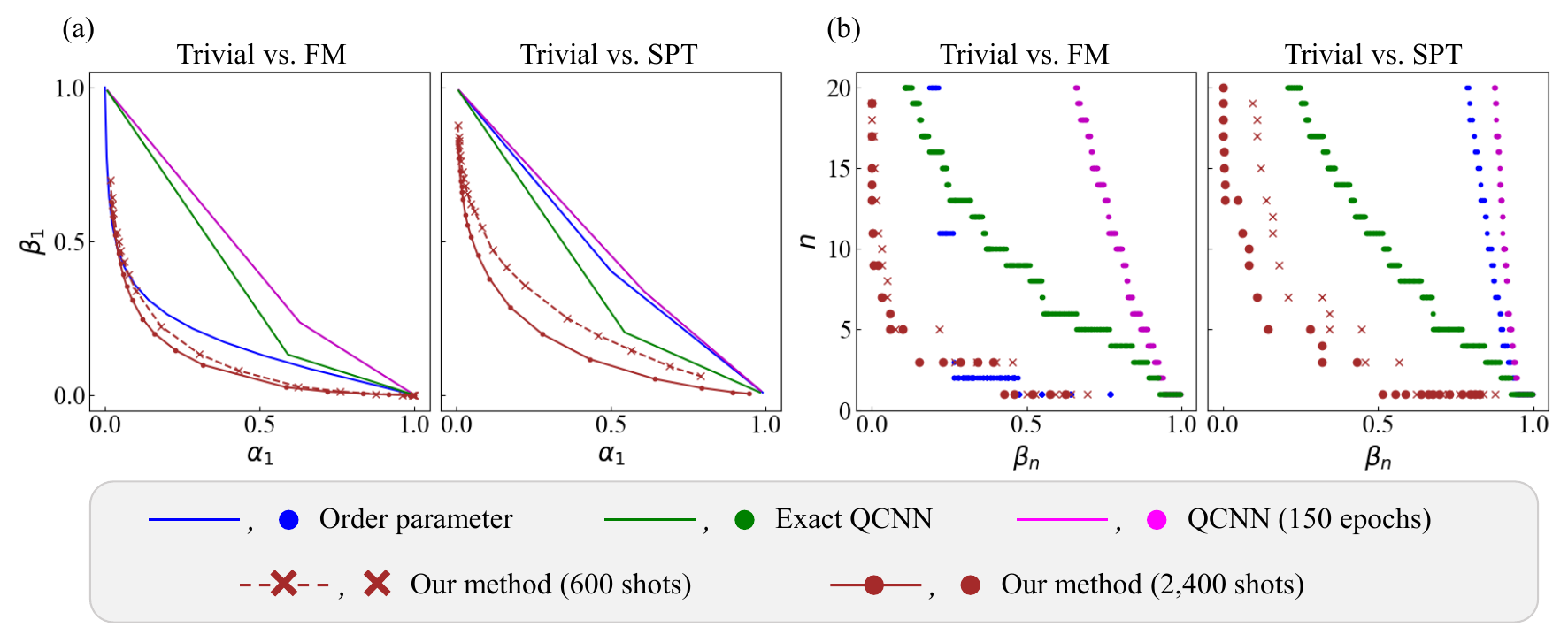}
    \caption{Type-\Romannumeral{1} and Type-\Romannumeral{2} error probabilities, $\alpha_n$ and $\beta_n$, for the test data in the order parameter, Exact QCNN, QCNN, and our method on $L=27$ qubits in the Trivial vs. FM case and Trivial vs. SPT case. 
    Panel (a) shows the error probabilities $\alpha_{1}$ and $\beta_{1}$ for a single-copy test dataset ($n = 1$), whereas Panel (b) shows the number of test data copies $n$ required to achieve $\beta_n$ under the condition $\alpha_n \leq 5\,\%$.}
    \label{fig:Error probabilities}
\end{figure*}

\subsection*{Training costs}
\label{sec:Training costs}

We present the validation loss, expressed as the Mean Squared Error (MSE) loss between the test data and validation labels with respect to the number of training data copies, evaluated across several cases:
\begin{equation}\label{eq:validation loss}
\text{MSE} = \frac{1}{N_{\text{test}}}\sum^{N_{\text{test}}}_{i=1}\qty(f(\rho^{(i)}_{\text{test}}) - y^{(i)}_{\text{test}})^2,
\end{equation}
where $y^{(i)}_{\text{test}} \in \{0,1\}$ is the validation label for the test data $\rho^{(i)}_{\text{test}}$, and where $f(\rho^{(i)}_{\text{test}})$ denotes the probability that $\rho^{(i)}_{\text{test}}$ is classified as validation label $y^{(i)}_{\text{test}}=1$.

%\paragraph*{Comparison with the QCNN.}
%{\it Comparison with the QCNN.} 
We compare our method with the QCNN, in terms of the number of training data copies. 
In our method, the probability of classifying a test label $y^{(i)}_{\text{test}}=1$ is set as $f(\rho^{(i)}_{\text{test}})$ in Eq.~(\ref{eq:validation loss}), using $a=0$ to give equal weight to both labels. 
For the QCNN, the output expectation value is used as $f(\rho^{(i)}_{\text{test}})$. 
The resultant learning curves for a 15-qubit system ($L = 15$) are shown in Fig.~\ref{fig:Comparison with QCNNs}. 
Note that each of the four panels has different scales and ranges for the horizontal and vertical axes. 
Note also that the QCNN uses 1,000 copies of training data per expectation value estimation in this experiment (i.e., $T_{\text{est}} = 1{,}000$).
More precisely, Panel (a) shows results for our method, with training shots calculated as $N_{\text{train}} \times T_{\text{state}} = 20 \times T_{\text{state}}$, shown on an axis in units of $10^3$. 
Panel (b) shows results for the QCNN, with training shots calculated as $N_{\text{train}} \times T_{\text{est}} \times N_{\text{epoch}} \times 2 = 20 \times 1{,}000 \times N_{\text{epoch}} \times 2 = 40{,}000 \times N_{\text{epoch}}$, shown on an axis in units of $10^6$. 
These results show that our method achieves lower validation losses while using fewer than one-thousandth of the training data copies required by the QCNN.
This significant improvement is attributed to the fact that gradient-based variational training methods generally require a large number of training data copies, which is not needed for our method. 

\begin{figure*}[htp!]
    \centering
    \includegraphics[width=1\columnwidth]{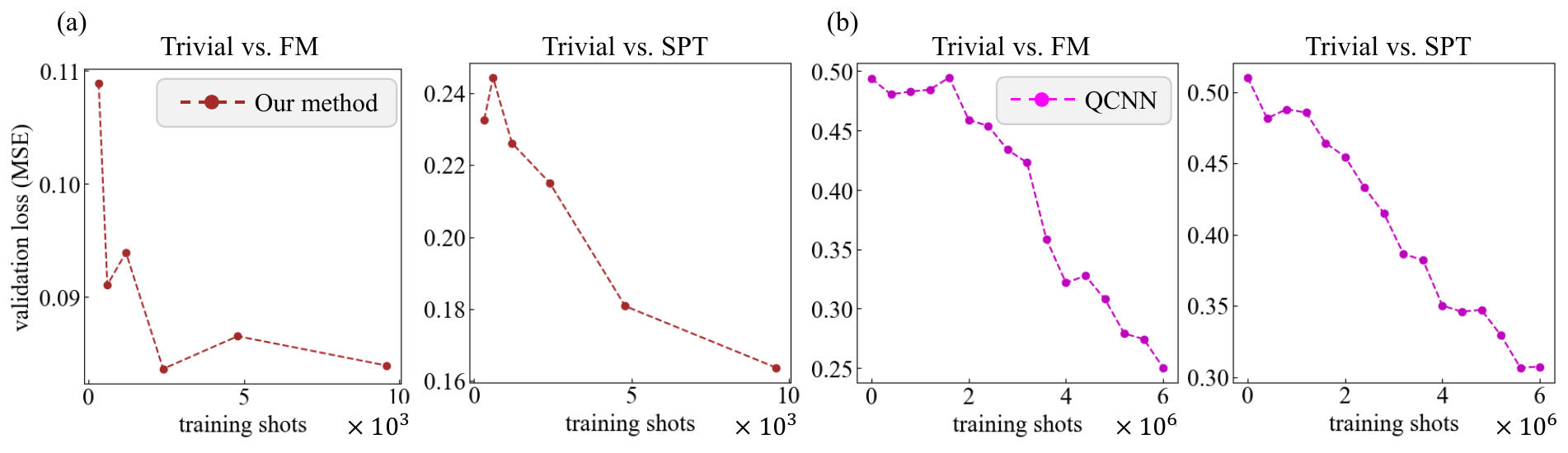}
    \caption{Learning curves for our method and the QCNN on $L = 15$ qubits in the Trivial vs. FM case and Trivial vs. SPT case. 
    Panels (a) and (b) show the results for our method and the QCNN, respectively. 
    The training shots represent the total number of training data, and the validation loss (MSE) is defined in Eq.~(\ref{eq:validation loss}). 
    }
    \label{fig:Comparison with QCNNs}
\end{figure*}

%\paragraph*{Comparison with the low-weight QCNN.}
%{\it Comparison with the low-weight QCNN.} 

Next, we compare our method with the low-weight QCNN, in terms of the required number of training data. 
The learning curves for a 15-qubit system ($L = 15$) for both our method and the low-weight QCNN are shown in Fig.~\ref{fig:Comparison with low-weight QCNNs}; 
the results for our method are the same as those in Fig.~\ref{fig:Comparison with QCNNs}, while the total training shots of the low-weight QCNN is calculated as $N_{\text{train}} \times T_{\text{state}} = 20 \times T_{\text{state}}$. 
This figure indicates that our method achieves a lower validation loss compared to the low-weight QCNN with the same number of training data. 
Furthermore, the low-weight QCNN requires fewer training data than the QCNN shown in Fig.~\ref{fig:Comparison with QCNNs} to reach comparable validation loss levels.

\begin{figure}[htp!]
    \centering
    \includegraphics[width=0.8\columnwidth]{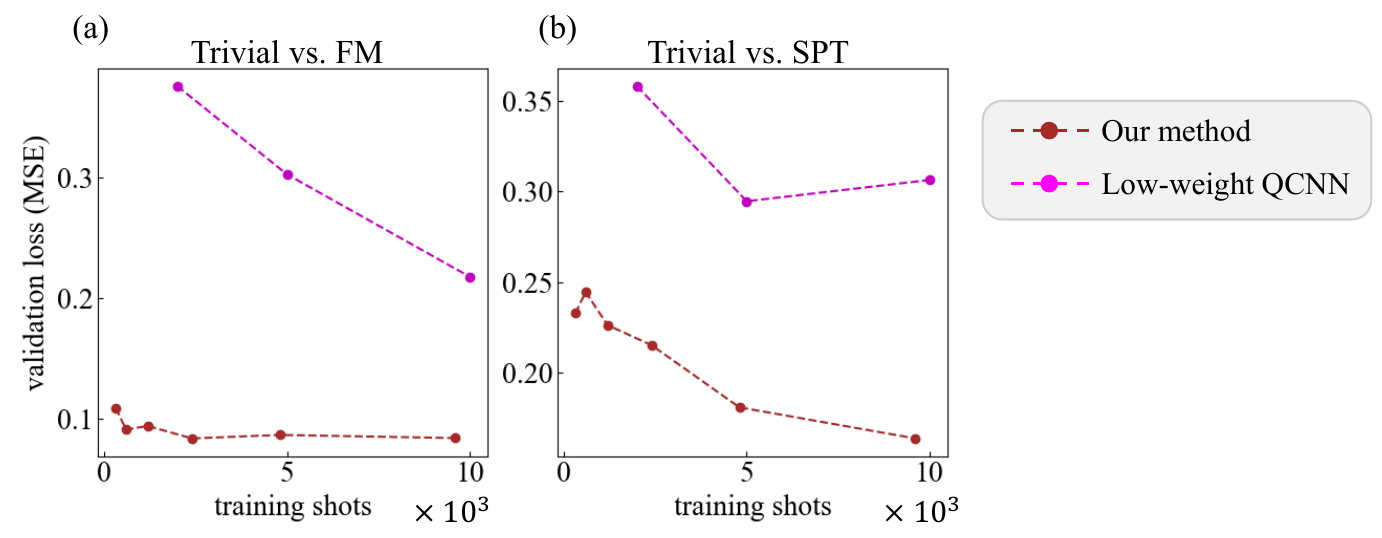}
    \caption{Learning curves for our method and the low-weight QCNN on $L = 15$ qubits. 
    Panel (a) corresponds to the Trivial vs. FM case, and Panel (b) to the Trivial vs. SPT case; both are shown in units of $10^3$. 
    The results for our method are the same as those in Fig.~\ref{fig:Comparison with QCNNs}.}
    \label{fig:Comparison with low-weight QCNNs}
\end{figure}

We also compare our method and the low-weight QCNN, in terms of the classical computational time complexity required to construct the measurement process. 
In our method, the system size $L$ is divided into groups of $k$ qubits. 
For each group, classical processing involves taking an ensemble average of matrices obtained from partial tomography (Training step 3) and performing an eigendecomposition to construct the gate implementation for the quantum Neyman-Pearson test (Training step 4). 
This classical processing is repeated $L/k$ times for each group. 
The ensemble averaging for a $k$-qubit matrix naively requires $\order{2^{2k}}$ time, eigendecomposition requires $\order{2^{3k}}$ time~\cite{Francis1961,Francis1962}, and constructing a general unitary gate implementation naively requires $2^{\order{k}}$ time~\cite{Dawson2006}. 
Thus, the classical computational time complexity of our method is $\frac{L}{k} \times \qty(\order{2^{2k}} + \order{2^{3k}} + 2^{\order{k}}) = 2^{\order{k}}L$.
On the other hand, simulating a QCNN of system size $L$ (typically with $\order{\log(L)}$ depth~\cite{Cong2019}) up to weight $k'$ for the low-weight QCNN requires $\order{L^{k'} \log(L)}$ time for each expectation value calculation~\cite{Angrisani2025}. 
Our method therefore has a lower classical computational time complexity than the low-weight QCNN; note that, thus, it does not imply an advantage as the quantum algorithm in all aspects. 
For example, if constant $k$ and $k'$ are chosen, the low-weight QCNN runs in $\order{\text{poly}(L) \log(L)}$ time for each expectation value calculation, whereas our method runs in $\order{L}$ time.

\subsection*{Scalability}

Finally, we examine the scalability of our method by comparing its Type-\Romannumeral{1} and Type-\Romannumeral{2} error probabilities of the test data for different system sizes $L$. 
The number of copies for the quantum Neyman-Pearson test is fixed to $n=1$.
The results for systems with $L = 15$, 27, and 81 qubits are shown in Fig.~\ref{fig:Scalability}. 
For each system size, only 600 copies of training data are used, and the other settings such as the hyperparameter $a$ are identical to those for our method in Fig.~\ref{fig:Error probabilities}.

Fig.~\ref{fig:Scalability} suggests that our method achieves lower error probabilities as the system size increases. 
This trend probably stems from the fact that quantum phase transitions become clearer in larger spin chains, as suggested by the definition of order parameters $o = \lim_{L\rightarrow\infty} \langle O \rangle$. 
Intuitively, while partial tracing over all but small groups of qubits would lead to a nearly maximally mixed state, in practice, information appears to remain in each small group even for large systems, such as the $L = 81$ qubits case. 
Determining which types of quantum phases exhibit this property remains a non-trivial question and is one of the key challenges in assessing the applicability of our method to complex and large-scale systems. 

\begin{figure}[htp!]
    \centering
    \includegraphics[width=0.7\columnwidth]{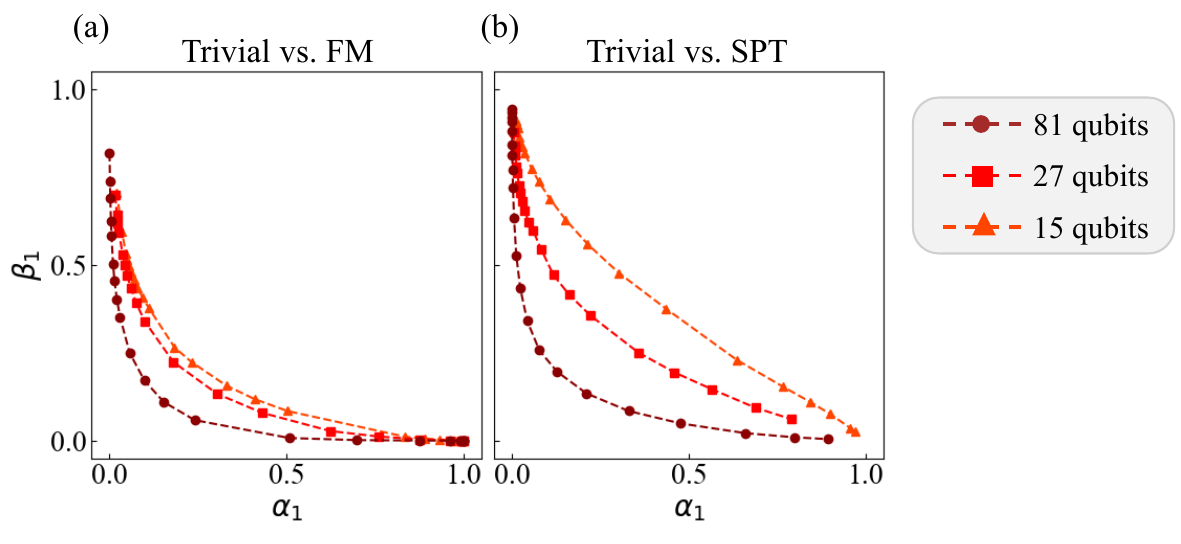}
    \caption{Type-\Romannumeral{1} and Type-\Romannumeral{2} error probabilities, $\alpha_1$ and $\beta_1$, for the test data in our method on $L = 15$, 27, 81 qubits. 
    Panel (a) corresponds to the Trivial vs. FM case, while Panel (b) represents the Trivial vs. SPT case.}
    \label{fig:Scalability}
\end{figure}

\section*{Summary and discussion}

This work has considered a scenario where quantum many-body states are available on quantum devices, either experimentally or through quantum state preparation algorithms. 
We have proposed a method to classify quantum phases by segmenting the quantum many-body system into small groups of qubits, performing partial tomography on each group, and constructing the quantum Neyman-Pearson test (the optimal strategy for distinguishing between two quantum states in quantum hypothesis testing) for each group.

We conducted a numerical analysis on quantum phase classification by performing partial tomography and utilizing the resulting tomographic data and the prior information about phases to construct measurements for classifying an unknown state. 
%The measurement used in this study approximates the quantum Neyman-Pearson test. 
Partial tomography is efficient with respect to the system size; however, it has the downside that it cannot acquire global information about a quantum state, which may hinder accurate phase classification.
Nevertheless, this study numerically demonstrates that leveraging partial tomographic data effectively enables a more efficient classification of some quantum phases in many-body quantum states compared to existing methods.
% \red{that have been successfully classified by existing methods\cite{Cong2019,Sander2025,Bermejo2024,Caro2022}}
In particular, we observe the following results: 
our method yields lower error probabilities than the QCNN with fewer copies of training data. 
Furthermore, without prior knowledge of the quantum phase, our approach attains lower error probabilities than the Exact QCNN and the order parameter. 
In comparison with the low-weight QCNN, our method achieves lower error probabilities with a similar number of training data and improves the required classical processing time from polynomial-log factor in system size $L$ to linear time. 
Our approach demonstrates adequate performance for large systems of up to 81 qubits. Moreover, we expect that its effectiveness will extend to a broader range of quantum phases, including those examined in Refs.~\cite{Cong2019,Huang2022manybody,Bermejo2024,Liu2023,Monaco2023}.

Our future challenges include investigating several fundamental aspects of our method as a machine learning model in greater detail. 
Specifically, we aim to explore generally achievable error probabilities for quantum phase classification (generalization), the methods of partial tomography that reduce error probabilities with fewer training data copies (trainability), and the effectiveness when applied to larger quantum many-body systems with more complex quantum phases (expressivity). 
Regarding expressivity, our method may be limited to classifying quantum phases that can be distinguished by low-body observables with respect to quantum states. 
This limitation is likely similar to the scope of applicability of the classical shadow-based method described in Refs.~\cite{Huang2022manybody,Bermejo2024}, where classification is efficient when phases are identifiable through low-body information. 
This reflects a broader challenge in QML: approaches such as classical shadows, QCNNs, and our method may be generally not equipped to efficiently capture nonlocal features, such as those arising in phases with long-range entanglement or topological order.   
In QML tasks beyond quantum phase classification, there is a need for algorithms capable of learning global observables to address nonlocal correlations.   
Relatedly, recent developments in QML have encountered obstacles, such as barren plateaus, and we suggest using quantum experiments and classical processing as one approach to overcome this barrier. 
To advance these efforts and further promote the development of QML, we will continue research in the hybrid domain of QML and QHT.

\section*{Data availability}
The datasets generated during and analysed during the current study are available in the following GitHub repository: \url{https://github.com/Tanji-A/Quantum-phase-classification-via-quantum-hypothesis-testing}.

\bibliography{ref}

% \noindent LaTeX formats citations and references automatically using the bibliography records in your .bib file, which you can edit via the project menu. Use the cite command for an inline citation, e.g.  \cite{Hao:gidmaps:2014}.

% For data citations of datasets uploaded to e.g. \emph{figshare}, please use the \verb|howpublished| option in the bib entry to specify the platform and the link, as in the \verb|Hao:gidmaps:2014| example in the sample bibliography file.

\section*{Acknowledgements}

This work is supported by MEXT Quantum Leap Flagship Program Grants No. JPMXS0118067285 and No. JPMXS0120319794 and JST Grant No. JPMJPF2221.

% Acknowledgements should be brief, and should not include thanks to anonymous referees and editors, or effusive comments. Grant or contribution numbers may be acknowledged.

\appendix
% \section*{Additional information}
\section{Settings of numerical simulations}
We first describe how phase classification was performed using the methods employed in the numerical results presented in the main text.

\paragraph*{Order parameter.}
The classification using the order parameter, detailed in Section \ref{app:order and cNeyman}, involves applying projective measurements defined by the observables $O_{\text{FM}}$ for the Trivial vs. FM case and $O_{\text{SPT}}$ for the Trivial vs. SPT case in the main text. 
These measurements are performed on the test data, and are utilized for testing whether the expectation values of the order parameter are zero or not, using the Bayesian test for the Trivial vs. FM case and the classical Neyman-Pearson test for the Trivial vs. SPT case. 

\paragraph*{Exact QCNN.}
The Exact QCNNs used in our numerical simulation are shown in Fig.~\ref{fig:Exact QCNN circuits}. The Exact QCNN for the FM phase was developed in this study, whereas the one for the SPT phase was taken from Ref.~\cite{Cong2019}. 
Both were constructed based on theoretical insights, such as order parameters and the renormalization group.
For classification using the Exact QCNN, also detailed in Section \ref{app:QCNN and cNeyman}, the test data are input into the circuit shown in Fig.~\ref{fig:Exact QCNN circuits}(a) for the Trivial vs. FM case and in Fig.~\ref{fig:Exact QCNN circuits}(b) for the Trivial vs. SPT case. 
The output expectation values are then evaluated using the classical Neyman-Pearson test to test whether they exceed $0.5$.
Test data with output values above $0.5$ are classified as label $y^{(i)}=1$, and those not exceeding $0.5$ as label $y^{(i)}=0$.
\begin{figure}[htp!]
    \centering
    \includegraphics[width=0.5\columnwidth]{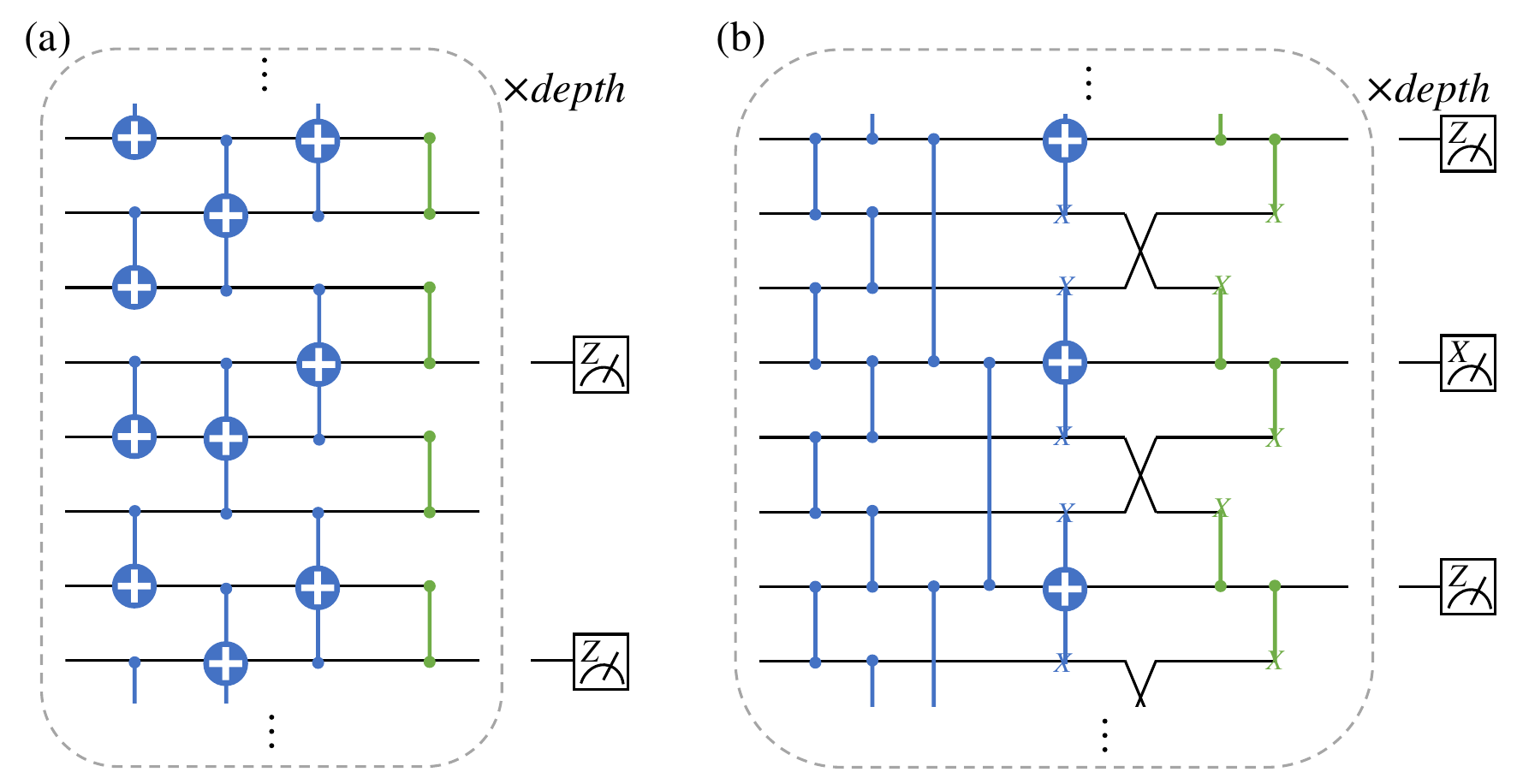}
    \caption{Exact QCNN circuits for (a) the FM and (b) SPT phases. 
    The details of (a) are provided in Section \ref{app:QCNN and cNeyman}, while (b) is proposed in Ref.~\cite{Cong2019}. 
    The circuit consists of convolutional layers (blue) and pooling layers (green, reduce the qubits), repeated for a specified depth, followed by Pauli measurements on the remaining qubits.}
    \label{fig:Exact QCNN circuits}
\end{figure}

\paragraph*{QCNN.}
The QCNN is trained using the ansatz shown in Fig.~\ref{fig:QCNN circuit}, initialized randomly. 
The training data are input, and the circuit parameters $\boldsymbol{\theta}$ are updated over epochs using the SPSA optimizer to minimize the MSE loss 
\begin{equation}
\text{MSE}(\boldsymbol{\theta}) = \frac{1}{N_{\text{train}}}\sum^{N_{\text{train}}}_{i=1}\qty(f(\boldsymbol{\theta},\rho^{(i)}) - y^{(i)})^2,
% \text{MSE}(\boldsymbol{\theta}) = \sum^{N_{\text{train}}}_{i=1}\left(\text{Tr}[\mathcal{E}_{\boldsymbol{\theta}}(\rho^{(i)})X] - y^{(i)}\right)^2,
\end{equation}
where $f(\boldsymbol{\theta}, \rho^{(i)})$ denotes the output expectation value when the training data $\rho^{(i)}$ is input into the QCNN parameterized by $\boldsymbol{\theta}$. 
This training is performed for both the Trivial vs. FM and the Trivial vs. SPT cases.
The trained QCNN classifies the test data using the classical Neyman-Pearson test, following the same procedure as described for the Exact QCNN: output values above $0.5$ are assigned label $y^{(i)}=1$, and those not exceeding $0.5$ are assigned label $y^{(i)}=0$.
\begin{figure}[htp!]
    \centering
     \includegraphics[width=0.4\columnwidth]{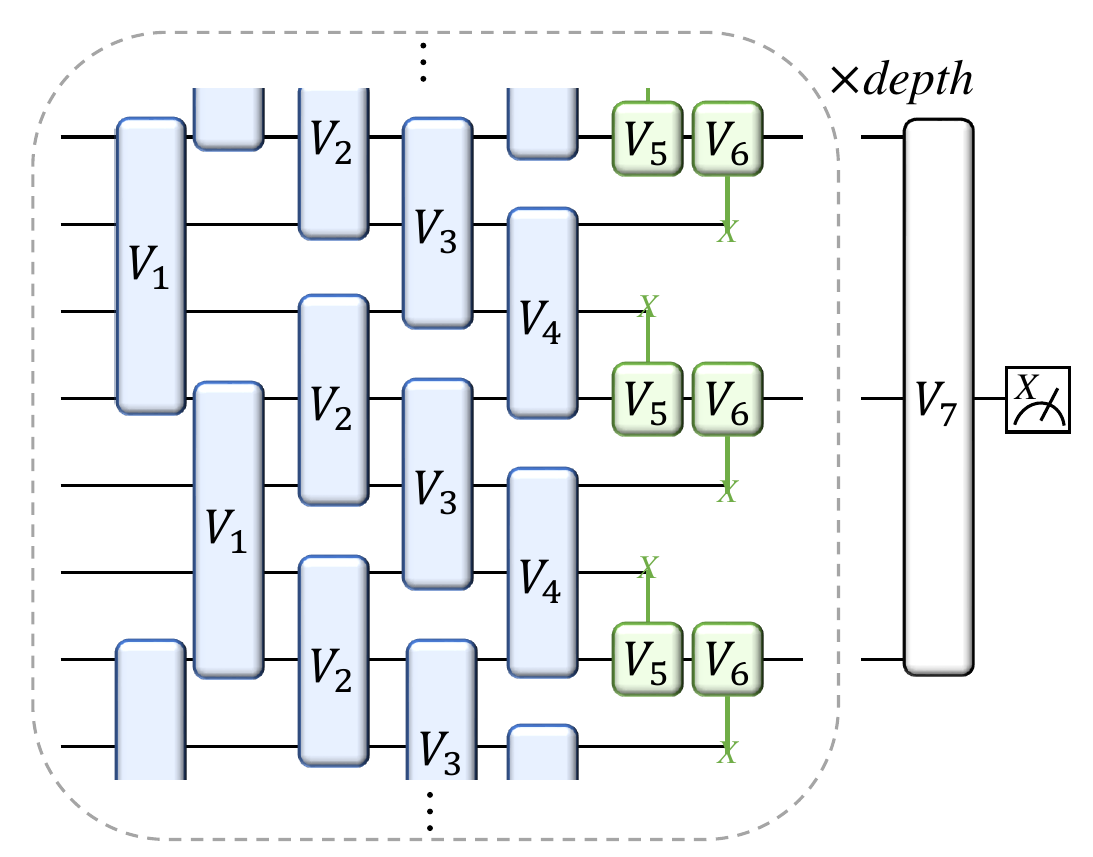}
    \caption{QCNN ansatz proposed in Ref.~\cite{Cong2019}.
    The circuit consists of convolutional layers (blue) and pooling layers (green, reduce the qubits), repeated for a specified depth, followed by fully connected layer (black) and Pauli measurements on the remaining qubits.
    The unitaries are parameterized as $V = \mathrm{exp}(-i\sum_j \theta_j\Lambda_j)$, where $\{\Lambda_j\}$ are generalized Gell-Mann matrices and $\{\theta_j\}$ are real parameters.}
    \label{fig:QCNN circuit}
\end{figure}

\paragraph*{low-weight QCNN.}
The low-weight QCNN uses classical shadow snapshots of the training data obtained through quantum experiments and is trained by classical simulation using only the low-weight components of the Heisenberg-evolved observables, i.e., Pauli basis elements with a small number of non-identity qubits.
Since the snapshots are collected once and can be reused for any number of epochs, the low-weight QCNN is trained for 300 epochs, where the validation loss reaches a stable level.
The low-weight QCNN limit the weight of Pauli basis elements up to three, employs random Pauli measurements for classical shadow, and the other settings such as the ansatz and loss function follow the same as the QCNN.

\section{Majority vote}\label{sec:majority vote}
After performing the partitioned quantum Neyman--Pearson tests on each subsystem, the individual subsystem predictions must be combined by a classical post-processing step. 
In our numerical simulations we use a simple majority vote. 
In this section we justify this choice by showing that, under an exponential clustering assumption, majority vote reduces the prediction variance. 
We also discuss alternative classical aggregation methods.

The following assumption is known to hold, e.g., for unique ground states of local and gapped Hamiltonians~\cite{Hastings2006,Nachtergaele2006}. 
It is a typical property in simple models but may fail at critical points where phase transitions occur.
Consequently, while it supports our design choice, it does not fully cover all regimes probed in our simulations.

% \begin{assumption}[Exponential clustering~\cite{Hastings2006,Nachtergaele2006}]
% \label{asm:clustering}
% Consider two disjoint subsystems $\mathcal{H}_X,\mathcal{H}_Y$ supported on site-sets $X,Y$, respectively. 
% Let $\rho_X,\rho_Y,\rho_{XY}$ denote the reduced density matrices (RDMs) of a \blue{\sout{global} entire} state on $\mathcal{H}_X$, $\mathcal{H}_Y$, and $\mathcal{H}_X\otimes\mathcal{H}_Y$. 
% \red{[ Does this mean the following? ``Let $\rho_X$ and $\rho_Y$ denote the reduced density matrices (RDMs) of a global state $\rho_{XY}$ on $\mathcal{H}_X\otimes\mathcal{H}_Y$."  ]} 
% Let $r_{XY} = \min_{x\in X,\,y\in Y}\mathrm{dist}(x,y)$ be the graph-theoretic distance between $X$ and $Y$, and let $\xi>0$ denote the correlation length. For any local observables $A_X$ on $\mathcal{H}_X$ and $B_Y$ on $\mathcal{H}_Y$ with $\|A_X\|\le 1$ and $\|B_Y\|\le 1$, correlations decay exponentially with distance,
% \begin{equation}
% \bigl|\mathrm{Tr}\!\bigl[(A_X\otimes B_Y)\,(\rho_{XY}-\rho_X\otimes\rho_Y)\bigr]\bigr|
% \;\le\; \mathrm{Const.}\times e^{-r_{XY}/\xi}\,.
% \end{equation}
% \end{assumption}

\begin{assumption}[Exponential clustering~\cite{Hastings2006,Nachtergaele2006}]
\label{asm:clustering}
Let $\rho$ be the density operator of the entire system.
For disjoint subsystems $\mathcal{H}_X,\mathcal{H}_Y$ supported on site-sets $X,Y$, define the reduced density matrices by $\rho_X = Tr_{\lnot X} ( \rho )$, $\rho_Y = Tr_{\lnot Y} ( \rho )$, and $\rho_{XY} = Tr_{\lnot (X \cup Y)} ( \rho )$, where $\lnot X$ denotes the complement of $X$ in the entire system.
Let $r_{XY} = \min_{x\in X,\,y\in Y}\mathrm{dist}(x,y)$ be the graph-theoretic distance between $X$ and $Y$, and let $\xi>0$ denote the correlation length. For any local observables $A_X$ on $\mathcal{H}_X$ and $B_Y$ on $\mathcal{H}_Y$ with $\|A_X\|\le 1$ and $\|B_Y\|\le 1$, correlations decay exponentially with distance,
\begin{equation}
\bigl|\mathrm{Tr}\!\bigl[(A_X\otimes B_Y)\,(\rho_{XY}-\rho_X\otimes\rho_Y)\bigr]\bigr|
\;\le\; \mathrm{Const.}\times e^{-r_{XY}/\xi}\,.
\end{equation}
\end{assumption}

Treating each partitioned quantum Neyman--Pearson test as a weak classifier, the overall procedure can be viewed as a form of model ensemble learning.
In classical machine learning, a simple majority vote is well known to be a natural ensemble method for classification problems, and its effectiveness is governed by the correlations among weak classifiers. 
Under Asm.~\ref{asm:clustering}, the covariance between the $i$- and $j$-th subsystems decays exponentially with their separation $r_{ij}$ and the inverse correlation length $1/\xi$. 
As a result, when aggregating their predictions by majority vote, the estimator exhibits variance reduction analogous to model ensembling methods.
In particular, when $\xi$ is sufficiently small, the variance scales as $\mathcal{O}(1/M)$ with $M=L/k$ the number of subsystems.

\begin{theorem}[Variance reduction by the majority vote on a $d$-dimensional model]
\label{thm:majority vote variance}
Let the $M=L/k$ subsystems be arranged on a $d$-dimensional lattice, and let
$\{S_j^{(0)},S_j^{(1)}\}$ be a two-outcome POVM on subsystem $j\in\{1,\dots,M\}$.
For each $j$, define a random variable $Y_j\in\{0,1\}$ whose marginal distribution is Bernoulli with $\Pr(Y_j=1)=\mathrm{Tr}(S_j^{(1)}\rho_j)$, where $\rho_j$ is the RDM on subsystem $j$.  
Let the average estimator be
$\bar{Y}=\frac{1}{M}\sum_{j=1}^{M}Y_j$, which serves as the majority vote estimator when thresholded at 0.5.
Under Asm.~\ref{asm:clustering}, 
\begin{equation}
\mathrm{Var}(\bar{Y})
\;\le\; \frac{\sigma^2}{M}\left[
1+\frac{C}{\sigma^2}\left\{\left(\frac{1+e^{-1/\xi}}{1-e^{-1/\xi}}\right)^{\!d}-1\right\}
\right],
\label{eq:mv bound}
\end{equation}
where $\sigma^2=\max_j\mathrm{Var}(Y_j)$ and $C$ is the maximum constant appearing in Asm.~\ref{asm:clustering} taken over all subsystems.
In particular, for small correlation length $\xi$, the dominant scaling is
$\mathrm{Var}(\bar{Y})\lessapprox \sigma^2/M$.
\end{theorem}

\begin{proof}
Write $\sigma_j^2=\mathrm{Var}(Y_j)$ and
$\mathrm{cov}_{i,j}=\mathrm{Cov}(Y_i,Y_j)$.  By Asm.~\ref{asm:clustering} and
$\|S_j^{(1)}\|\le 1$,
\begin{equation}
|\mathrm{cov}_{i,j}|
= \bigl|\mathbb{E}[Y_iY_j]-\mathbb{E}[Y_i]\mathbb{E}[Y_j]\bigr|
= \bigl|\mathrm{Tr}\!\bigl[(S_i^{(1)}\!\otimes S_j^{(1)})(\rho_{ij}-\rho_i\otimes\rho_j)\bigr]\bigr|
\le \mathrm{Const.}\times e^{-r_{ij}/\xi}\,,
\end{equation}
and thus
\begin{equation}
\mathrm{Var}(\bar{Y})
=\frac{1}{M^2}\Bigl(\sum_{j=1}^{M}\sigma_j^2
+2\sum_{1\le i<j\le M}\mathrm{cov}_{i,j}\Bigr)
\;\le\; \frac{\sigma^2}{M}+\frac{2C}{M^2}\sum_{1\le i<j\le M}e^{-r_{ij}/\xi}.
\end{equation}
For the $d$-dimensional lattice with $L_1$ distance, the generating function identity
\begin{equation}
\sum_{r=1}^{\infty} N_d(r)\, q^r
= \left(\frac{1+q}{1-q}\right)^{d} - 1,
\end{equation}
holds for any scalar $q$ with $|q|<1$, where $N_d(r)$ denotes the number of sites at distance $r$ from a given site (see, e.g., Thm.~2.7 of ~\cite{Beck2015}, from which it follows by a straightforward calculation).
Using this,
\begin{equation}
\frac{1}{M^2}\sum_{1\le i<j\le M}e^{-r_{ij}/\xi}
\;\le\; \frac{1}{M^2}\sum_{r=1}^{\infty}e^{-r/\xi}\frac{1}{2}\sum_{i=1}^M N_d(r) = \frac{1}{2M}\sum_{r=1}^{\infty}N_d(r)e^{-r/\xi},
\end{equation}
where the factor $1/2$ accounts for the fact that each unordered pair $(i,j)$ is counted twice when summing over all $i$.
Substituting this bound completes the proof of Eq.~\eqref{eq:mv bound}.  
\end{proof}
Although majority vote is not necessarily optimal in all settings, under the exponential clustering assumption it behaves as a standard and justified ensemble method.
Finally, we note that the majority vote is obtained simply by thresholding the average estimator $\bar{Y}$ at 0.5.
Thus, the variance reduction of $\bar{Y}$ directly justifies the stability of the majority vote.
Moreover, in the extreme uncorrelated case (i.e., $\forall i,j,i\neq j,\ \mathrm{cov}_{i,j}=0$) for odd $M$, the majority vote estimator $X\sim\mathrm{Bernoulli}(\Pr\{\bar{Y}>0.5\})$ satisfies $\bigl|\mathbb{E}[X]-0.5\bigr| \;\ge\; \bigl|\mathbb{E}[\bar{Y}]-0.5\bigr|$, which follows from the ultra log-concavity of Poisson binomial distributions~\cite{Liggett1997,Johnson2013}.
This provides a futher rationale for adopting the majority rule.
We note that the above result concerns the variance of the estimator given fixed system and subsystem sizes $L$ and $k$.
If one instead increases $M$ by decreasing $k$, each subsystem maintains less information about the entire state, which may affect classification accuracy.
This trade-off is separate from the variance reduction analyzed in the theorem.

Beyond the uniform majority vote, we may also consider a weighted variant
\begin{equation}
\bar{Y}_{\boldsymbol{w}}
=
\frac{1}{\|\boldsymbol{w}\|_1}\sum_{j=1}^{M} w_j\,Y_j.
\end{equation}
Such weights may be chosen, for example, to reduce finite-size effects by assigning smaller weights to subsystems near the boundary and larger weights to those in the bulk.
Alternatively, weights may be determined from subsystem performance estimates obtained from RDMs estimated on training data, which does not require additional copies of the quantum states.
When sufficient test copies are available, the weights may be adjusted based on 
confidence estimates obtained at inference time.

\section{Additional numerical simulations}
\label{app:Additional simulations}
Here, we present additional results to complement the numerical findings presented in the main text.

\subsection{Robustness under depolarizing noise}
To evaluate the robustness of our method under a standard noise model, we consider global depolarizing noise defined by
\begin{equation}
D_p(\rho) = (1 - p)\,\rho + p\,\frac{I}{d},
\end{equation}
where $\rho$ is an arbitrary quantum state, $p$ is the depolarizing probability, $d$ is the dimension of $\rho$, and $I$ is the identity operator. 
In the training step, depolarizing noise with total probability $p$ is applied both to the training states $\{\rho^{(i)}\}_{i=1}^{N_{\text{train}}}$ and to the tomography unitary $\bigotimes_{j=1}^{L/k}U_j$. 
That is, we perform measurements in the computational basis on the noisy states $D_p\qty(\,(\bigotimes_{j=1}^{L/k}U_j)\,\rho^{(i)}\,(\bigotimes_{j=1}^{L/k}U_j)^\dagger)$ to obtain snapshots for the partial state tomography. 
In the test step, depolarizing noise with total probability $p$ is applied both to the test states $\{\rho_{\text{test}}^{(i)}\}_{i=1}^{N_{\text{test}}}$ and to the unitary $\bigotimes_{j=1}^{L/k} [\ket*{\lambda_{j,1}}, \ldots, \ket*{\lambda_{j,2^k}}]^\dagger$
used in the partitioned quantum Neyman-Pearson test.

The error probabilities when our method is run on a 27-qubit system ($L=27$) under the depolarizing noise are shown in Fig.~\ref{fig:depolarizing noise}. 
From the panel (a), where only training noise is applied, we observe that the performance of our method deteriorates very little. 
A rough analysis suggests that this robustness arises because, when depolarizing noise acts on both $\rho$ and $\sigma$ in the POVM element of the quantum Neyman-Pearson test defined in the main text (with $a=0$ and $n=1$), we obtain
\begin{equation}
\{D_p(\rho) - D_p(\sigma) > 0\}
= \{((1-p)\rho + p\frac{I}{d}) - ((1-p)\sigma + p\frac{I}{d}) > 0\}
= \{(1-p)(\rho - \sigma) > 0\}
= \{\rho - \sigma > 0\},
\end{equation}
indicating that the POVM remains unchanged. 
Although the eigenvalues of $\rho - \sigma$ are scaled by a factor of $1 - p$, the sign discrimination of these eigenvalues is handled by classical processing in our method and is therefore largely unaffected.
On the other hand, some performance degradation is observed in the panel (b), where noise is applied in both the training and test steps, as expected given that the evaluation considers the error probabilities for a single copy of the test data.
This can be understood as follows. 
Let $\alpha_1^{(p)}$ and $\beta_1^{(p)}$ be the Type-\Romannumeral{1} and Type-\Romannumeral{2} error probabilities for a single copy under depolarizing noise with probability $p$ in the test step, and then, using the definitions of the error probabilities as given in the main text, we have
\begin{align}
\alpha_1^{(p)} &= (1-p) \alpha_1^{(p=0)} + p \Tr(\frac{I}{d}(I - M_1)) = (1-p) \alpha_1^{(p=0)} + p \qty(1 - \Tr\qty(\frac{I}{d} M_1)),\\
\beta_1^{(p)} &= (1-p) \beta_1^{(p=0)} + p \Tr\qty(\frac{I}{d}M_1),
\end{align}
for any classification method $M_1$, not just our method. 
Thus, the results in the panel (b) follow naturally from the definition of the error probabilities and the presence of depolarizing noise.
Such degradation could potentially be mitigated by improving the state preparation of test data or by employing error mitigation techniques.

\begin{figure}[htp!]
    \centering
    \includegraphics[width=1\columnwidth]{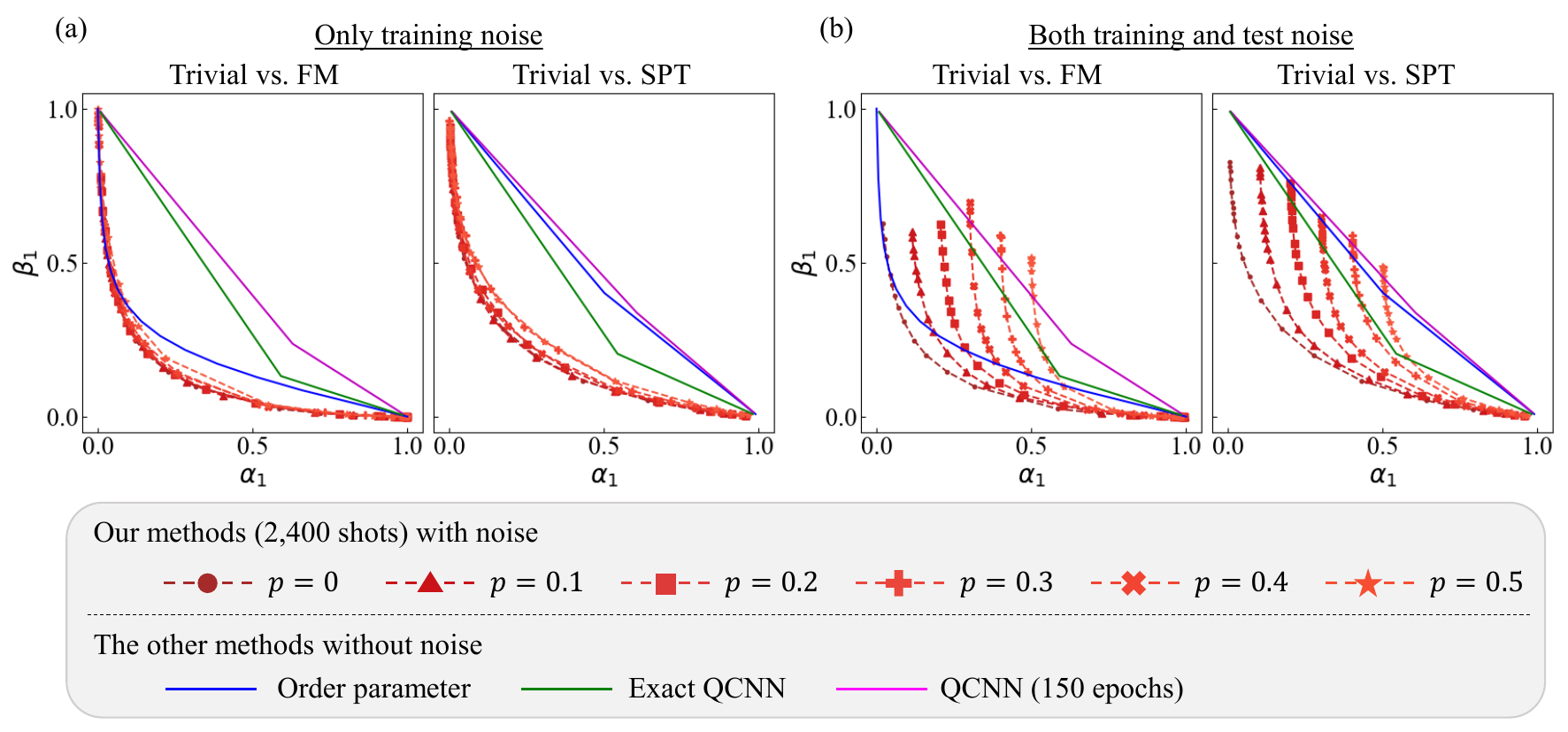}
    \caption{Type-\Romannumeral{1} and Type-\Romannumeral{2} error probabilities, $\alpha_1$ and $\beta_1$, under global depolarizing noise with probability $p$ for single-copy test data ($n=1$) on $L = 27$ qubits in the Trivial vs. FM and Trivial vs. SPT cases. The results for the other methods are identical to those in the main text, with no noise applied. 
    The panel (a) shows the results with noise applied only in the training step, and the panel (b) shows the results with noise applied in both the training and test steps.}
    \label{fig:depolarizing noise}
\end{figure}

\subsection{Error probabilities for $n_{\text{ent}} > 1$}
\label{app:n_ent>1}
The approximate quantum Neyman-Pearson test $\{S_j^{(0)}(n_{\text{ent}}, a), S_j^{(1)}(n_{\text{ent}}, a)\}$ in the main text involves an entangled measurement on $n_{\text{ent}}$ copies of quantum states. 
However, all numerical simulations in the main text were performed with $n_{\text{ent}} = 1$, as our method requires the eigendecomposition of $k n_{\text{ent}}$ qubits, which becomes exponentially difficult with increasing $n_{\text{ent}}$. 
This section demonstrates that increasing $n_{\text{ent}}$ beyond 1 provides limited benefits.

For $n_{\text{ent}} > 1$, the training steps 1 to 3 remain unchanged. 
In training step 4, the eigenvalues and eigenvectors of $\hat{\rho}_j^{\otimes n_{\text{ent}}} - e^{n_{\text{ent}}a} \hat{\sigma}_j^{\otimes n_{\text{ent}}}$ are calculated instead of $\hat{\rho}_j - e^a \hat{\sigma}_j$. 
In test step 2, measurements using the POVMs corresponding to the quantum Neyman-Pearson test are performed on $n_{\text{ent}}$ copies of the test data.

For $n_{\text{ent}} = 1$ and $3$, the Type-\Romannumeral{1} and Type-\Romannumeral{2} error probabilities for $n = 3$ copies of test data, $\alpha_3$ and $\beta_3$, are shown in Fig.~\ref{fig:n_ent g1}. 
In the partial tomography of training data, if an infinite number of copies are used, the resulting states are obtained by performing partial traces that retain only each $k$-qubit group. 
The infinite-shots results therefore correspond to the quantum Neyman-Pearson test constructed using these partial traced states for each $k$-qubit group. 
Fig.~\ref{fig:n_ent g1} indicates that there is no significant difference in error probabilities between $n_{\text{ent}} = 1$ and $3$, regardless of finite or infinite shots and whether the system size is $L = 15$ or $27$ qubits. 
However, for $n_{\text{ent}} = 3$, eigendecomposition and gate construction for the quantum Neyman-Pearson test are required for $k n_{\text{ent}} = 2\times3 = 6$ qubits, resulting in a significant increase in classical computational complexity.
While larger $n_{\text{ent}}$ might reduce error probabilities, the exponential increase in classical computational time with $n_{\text{ent}}$ makes $n_{\text{ent}} = 1$ the most practical choice.

\begin{figure}[htp!]
    \centering
    \includegraphics[width=0.7\columnwidth]{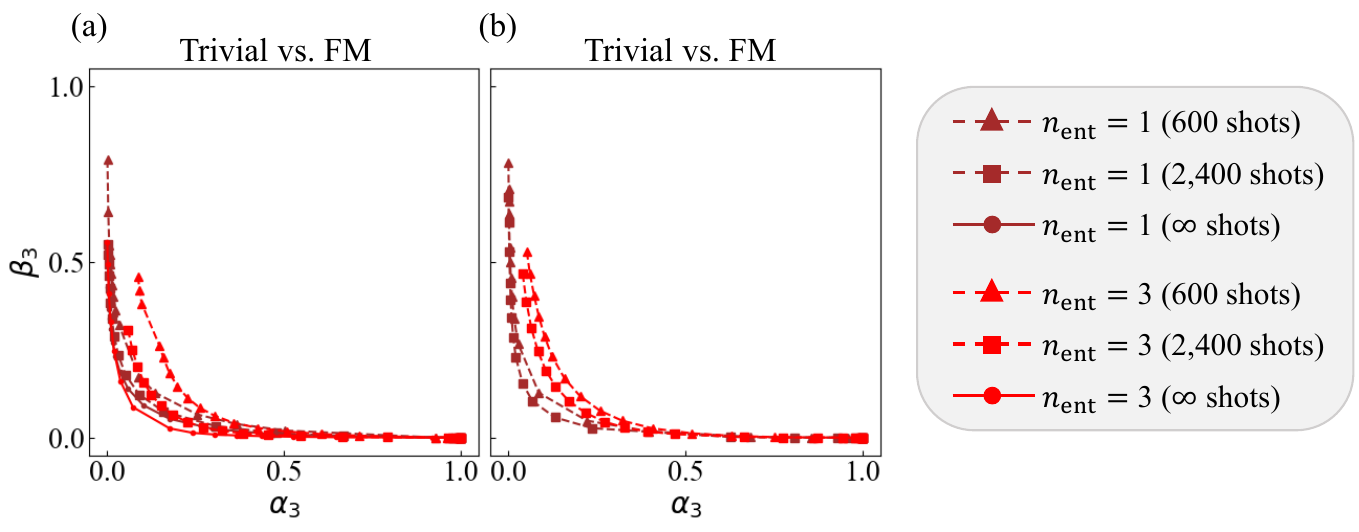}
    \caption{Type-\Romannumeral{1} and Type-\Romannumeral{2} error probabilities, $\alpha_3$ and $\beta_3$, for $n = 3$ copies of test data with $n_{\text{ent}} = 1$ and $3$. 
    Panel (a) results for an $L = 15$ qubit system, and panel (b) results for an $L = 27$ qubit system, both in the Trivial vs. FM case ($k = 2$). 
    The total number of training data copies is 600 or 2,400, with infinite-copy results also included for (a).}
    \label{fig:n_ent g1}
\end{figure}

\subsection{Two-dimensional model}
\label{app:2d model}

In this section, we demonstrate that our method could work effectively for two-dimensional models where the dividing strategy is non-trivial, even with a simple dividing method. 
We conduct numerical simulations for quantum phase classification of the two-dimensional Toric code Hamiltonian with magnetic fields
\begin{equation}
H = H_{\text{TC}} - h_X \sum_{i=1}^L X_i - h_Z \sum_{i=1}^L Z_i,
\end{equation}
where $X_i$ ($Z_i$) are the Pauli $X$ ($Z$) operators on the $i$-th qubit, and $h_X$ ($h_Z$) are the tunable strengths of the magnetic fields in the $X$ ($Z$) direction. 
Here, the Toric code Hamiltonian is
\begin{equation}
H_{\text{TC}} = - \sum_{p} A_p - \sum_{s} B_s,
\end{equation}
where $A_p = \prod_{i\in\text{plaquette}(p)} X_i$ and $B_s = \prod_{i\in\text{star}(s)} Z_i$ are the plaquette and star operators, respectively. 
The ground states of this Hamiltonian undergo a quantum phase transition between the topological phase of the Toric code and the trivial magnetic phase~\cite{Sander2025, Dusuel2011}. 
For instance, changing $h_Z$ at $h_X = 0$ induces a phase transition at $h_Z = 0.34$, and at $h_X = 0.1$, the phase transition occurs at $h_Z = 0.35$~\cite{Sander2025}.
We classify the quantum phases using the Exact QCNN and our method.
The test data, which are quantum states we want to phase classify, are common between the two methods and consist of 100 ground states near $h_X = 0.1$ and $h_Z = 0.35$. 
The training data required for our method are ground states along $h_X = 0$ with $h_Z$ evenly divided into 20 points in $[0, 0.68]$, with labels assigned as $y^{(i)} = 0$ for the trivial magnetic phase and $y^{(i)} = 1$ for the topological phase.
For numerical simulations, we use the MPS mapped to one dimension, as illustrated by the red dash-dot line in Fig.~\ref{fig:2d results}(a). 
Approximate ground states are prepared using the finite-size DMRG algorithm with a maximum bond dimension of 1500.

We describe the settings for each method.
\paragraph*{Exact QCNN.}
For classification using the Exact QCNN, detailed in Section \ref{app:QCNN and cNeyman}, the test data are input into the circuit proposed in Ref.~\cite{Sander2025}, and then the output expectation values are tested whether they exceed 0.5 using the classical Neyman-Pearson test. 
Test data with output values above 0.5 are classified as label $y^{(i)} = 1$, and those not exceeding 0.5 as label $y^{(i)} = 0$.
The limited system size in our numerical simulation ($L = 18$ qubits) restricts the circuit depth, which may hinder the Exact QCNN from achieving its expected performance.

\paragraph*{Our method.}
For our method, we use $k = 2$ qubit groups and the algorithm described in the Methods section of the main text. 
We divide the quantum many-body states into groups of $k = 2$ qubits based on the one-dimensional ordering shown in Fig.~\ref{fig:2d results}(a), where neighboring qubits are grouped sequentially from one end of the chain. 
Other settings, such as the partial tomography method, follow the configurations in the main text.

\vspace{3mm}
The Type-\Romannumeral{1} and Type-\Romannumeral{2} error probabilities, $\alpha_n$ and $\beta_n$, for the test data on $L = 18$ qubits are shown in Figs.~\ref{fig:2d results}(b)(c). 
These figures indicate that our method achieves sufficient performance even with simple dividing strategy. 
Note that the Exact QCNN is provided with prior knowledge of quantum phases but no training data, whereas our method uses training data without prior knowledge of the quantum phases, and this difference in problem setup precludes a direct comparison of performance.

\begin{figure}[htp!]
    \centering
    \includegraphics[width=1\columnwidth]{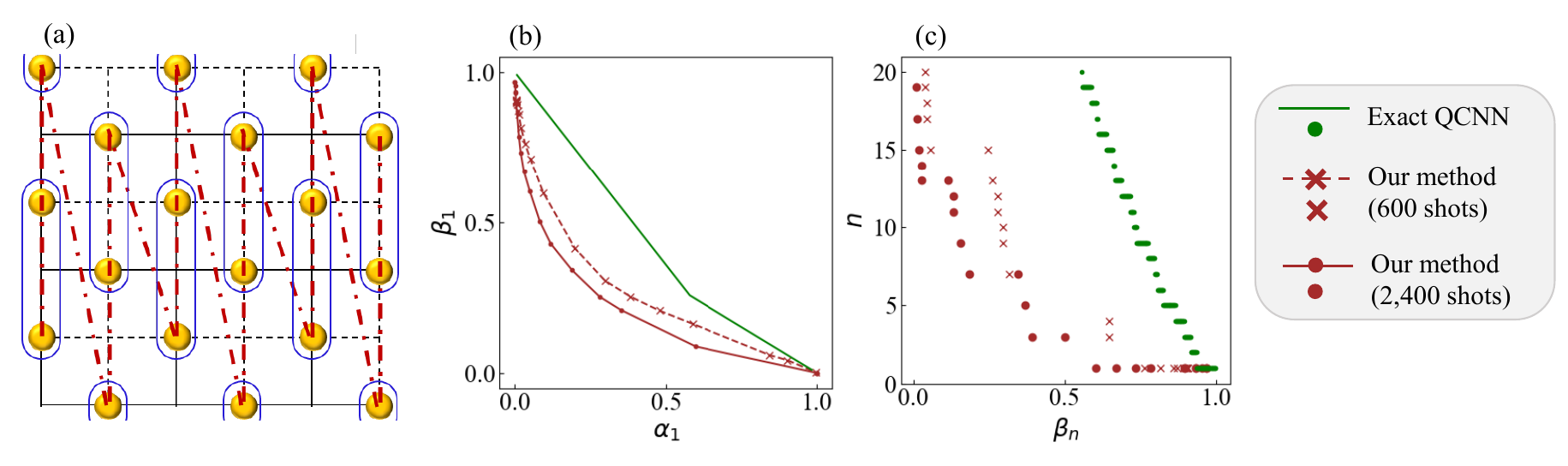}
    \caption{(a) Toric code lattice with $L = 18$ qubits. 
    The red dashed-dotted line indicates the mapping order to one dimension for the MPS representation, and the blue circles represent the dividing strategy used in our method. 
    (b)(c) Type-\Romannumeral{1} and Type-\Romannumeral{2} error probabilities, $\alpha_n$ and $\beta_n$, for the test data in the Exact QCNN and our method on $L=18$ qubits. 
    Panel (b) shows the error probabilities $\alpha_{1}$ and $\beta_{1}$ for a single-copy test dataset ($n = 1$), whereas panel (c) shows the number of test data copies $n$ required to achieve $\beta_n$ under the condition $\alpha_n \leq 5\,\%$. 
    The total number of training data copies used for our method is 600 or 2,400, while the Exact QCNN is the circuit proposed in Ref.~\cite{Sander2025}, which does not require training.
    It is possible that the Exact QCNN does not perform sufficiently due to the structural constraints on the $L=18$ qubit system size.}
    \label{fig:2d results}
\end{figure}

\section{Details of methods other than our method}
Here, we provide detailed explanations of the methods other than our method used for comparison in the numerical simulations presented in the main text.

\subsection{Classical hypothesis testing}

Classical hypothesis testing~\cite{Fisher1935,Neyman1933,Karlin1956,Lehmann2006} serves as a core approach in statistical inference, developed to evaluate competing hypotheses based on sample data. 
In this framework, two hypotheses are considered: the null hypothesis $H_0$, which represents a default state or baseline assumption, and the alternative hypothesis $H_1$, which represents an effect or deviation from the baseline. 
Classical hypothesis testing procedures are used to make decisions under uncertainty, particularly in fields like scientific research, quality control, and medical diagnostics, where making reliable decisions based on observed data is essential.

The decision-making process involves controlling error probabilities associated with incorrect decisions. 
Specifically, a Type-\Romannumeral{1} error occurs if we reject $H_0$ when it is true, with probability denoted by $\alpha$, and a Type-\Romannumeral{2} error occurs if we fail to reject $H_0$ when $H_1$ is true, with probability denoted by $\beta$.
Balancing these errors is central to classical testing; the significance level $\alpha$ is usually predefined, while the power of the test, $1-\beta$, represents the probability of correctly rejecting $H_0$ when $H_1$ is true. 

A fundamental result in hypothesis testing is the Neyman-Pearson lemma, which provides a criterion for constructing the most powerful test for simple hypotheses at a given significance level $\alpha$. 

\begin{lemma}[Neyman-Pearson Lemma~\cite{Neyman1933}]
\label{lemma:neyman-pearson}
Let $\theta$ be a parameter determining the probability distribution of a random variable $X$, which has a probability density $p(x \mid \theta)$.
Consider the simple null hypothesis $H_0: \theta = \theta_0$ and the simple alternative hypothesis $H_1: \theta = \theta_1$. 
Define the likelihood ratio as 
\begin{equation}
\Lambda(x) = \frac{p(x \mid \theta_1)}{p(x \mid \theta_0)}.
\end{equation}

A test $\phi$ of the form
\begin{equation}
\phi(x) = 
\begin{cases} 
1, & \mathrm{if}~\Lambda(x) > c \\ 
\gamma, & \mathrm{if}~\Lambda(x) = c \\
0, & \mathrm{if}~\Lambda(x) < c
\end{cases},
\end{equation}
exists such that it is the most powerful test at level $\alpha$, maximizing the power $1 - \beta$. 
Here, $c$ is the smallest constant satisfying $ \Pr\{\Lambda(X) > c \mid \theta = \theta_0\} \leq \alpha $, 
and $\gamma \in [0, 1]$ is chosen to satisfy $ \Pr\{\Lambda(X) > c \mid \theta = \theta_0\} + \gamma \Pr\{\Lambda(X) = c \mid \theta = \theta_0\} = \alpha $.
\end{lemma}

On the other hand, it is challenging to find a uniformly most powerful (UMP) test for all parameter values under composite hypotheses. 
However, under specific conditions such as a monotone likelihood ratio (MLR), it is possible to construct the UMP test. 
Here, the MLR is defined as a likelihood ratio such that the probability distribution differs for any $\theta_0<\theta_1$ and the ratio $p(x|\theta_1)/p(x|\theta_0)$ is a non-decreasing function of a real-valued function $T(x)$. 
The Neyman-Pearson lemma~\ref{lemma:neyman-pearson} is extended to composite hypotheses with MLR as the following lemma.

\begin{lemma}[\cite{Karlin1956}]
\label{lemma:karlin-rubin}
Let $\theta$ be a real parameter, and let $X$ be a random variable with a probability density $p(x \mid \theta)$ with MLR in a statistic $T(x)$. 
Consider the following hypotheses: 
\begin{align*}
\textup{(i)}~& H_0: \theta = \theta_{\mathrm{th}}, \quad H_1: \theta > \theta_{\mathrm{th}}\\
\textup{(ii)}~& H_0: \theta \leq \theta_{\mathrm{th}}, \quad H_1: \theta > \theta_{\mathrm{th}}.
\end{align*}

In both cases \textup{(i)} and \textup{(ii)}, there exists a UMP test $\phi$ of the form
\begin{equation}
\phi(x) = 
\begin{cases} 
1, & \mathrm{if}~T(x) > c\\
\gamma, & \mathrm{if}~T(x) = c\\
0, & \mathrm{if}~T(x) < c
\end{cases},
\end{equation}
at level $\alpha$, maximizing the power $1 - \beta$ for all $\theta$. 
Here, $c$ is the smallest constant satisfying $ \Pr\{T(X) > c \mid \theta = \theta_{\mathrm{th}}\} \leq \alpha $, 
and $\gamma \in [0, 1]$ is chosen to satisfy $ \Pr\{T(X) > c \mid \theta = \theta_{\mathrm{th}}\} + \gamma \Pr\{T(X) = c \mid \theta = \theta_{\mathrm{th}}\} = \alpha $.
\end{lemma}

\subsection{Order parameter}
\label{app:order and cNeyman}
\paragraph*{Classical Neyman-Pearson test using the order parameter of the SPT phase}

In the context of Trivial vs. SPT, the classification between the trivial phase and SPT phase is achieved by testing whether the expectation value of the order parameter $O_{\text{SPT}}$ is zero or not. 
To establish this, we first prove the following theorem (case (i) is applied in this subsection, while case (ii) is used in the next subsection) concerning the task of testing the expectation value of a Pauli string through projective measurements.

\begin{theorem}
\label{thm:UMP pauli string}
Let $O$ be a general Pauli string represented by the spectral decomposition 
\begin{equation}
O = (+1)\Pi_{+} + (-1)\Pi_{-},
\end{equation}
where $\Pi_{\pm}$ are projection operators onto the eigenspaces corresponding to eigenvalues $\pm1$.
For a given quantum state $\rho$, consider the projective measurement described by $O$, with probabilities $p(\pm1) = \Tr(\rho \Pi_{\pm})$. 
Suppose this measurement is performed $n$ times, and let $x$ denote the number of times the eigenvalue $+1$ is observed. Define $\ev{O} = \Tr(\rho O) = p(+1) - p(-1)$ as the expectation value of the Pauli string $O$. 
Consider the following hypotheses:
\begin{align*}
\textup{(i)}~& H_0: \ev{O} = \ev{O}_{\mathrm{th}}, \quad H_1: \ev{O} > \ev{O}_{\mathrm{th}}\\
\textup{(ii)}~& H_0: \ev{O} \leq \ev{O}_{\mathrm{th}}, \quad H_1: \ev{O} > \ev{O}_{\mathrm{th}}.
\end{align*}

In both cases \textup{(i)} and \textup{(ii)}, there exists a UMP test $\phi$ of the form
\begin{equation}\label{eq:test of thm:UMP pauli string}
\phi(x) = 
\begin{cases} 
1, & \mathrm{if}~x > c \\
\gamma, & \mathrm{if}~x = c \\
0, & \mathrm{if}~x < c
\end{cases},
\end{equation}
at level $\alpha$, maximizing the power $1 - \beta$ for all $\ev{O}$. 
Here, $c$ is the smallest non-negative integer satisfying $\Pr(x > c \mid \ev{O} = \ev{O}_{\mathrm{th}}) \leq \alpha $,
and $\gamma \in [0, 1]$ is chosen to satisfy $\Pr(x > c \mid \ev{O} = \ev{O}_{\mathrm{th}}) + \gamma \Pr(x = c \mid \ev{O} = \ev{O}_{\mathrm{th}}) = \alpha$.
\end{theorem}

\begin{proof}
The random variable $X$, which is the number of times that the eigenvalue $+1$ is observed in $n$ projective measurements described by $O$, follows a binomial distribution 
\begin{equation}
p(x \mid n, p_{\text{bin}}) = \binom{n}{x} p_{\text{bin}}^{x} (1 - p_{\text{bin}})^{n-x}, 
\end{equation}
where $p_{\text{bin}} = p(+1) = \frac{1 + \ev{O}}{2}$. 
In case (i), the hypotheses are rewritten as
\begin{equation}
H_0: p_{\text{bin}} = p_0 = \frac{1 + \ev{O}_{\mathrm{th}}}{2}, \quad 
H_1: p_{\text{bin}} = p_1 > \frac{1 + \ev{O}_{\mathrm{th}}}{2},
\end{equation}
and, in case (ii), the hypotheses are
\begin{equation}
H_0: p_{\text{bin}} = p_0 \leq \frac{1 + \ev{O}_{\mathrm{th}}}{2}, \quad 
H_1: p_{\text{bin}} = p_1 > \frac{1 + \ev{O}_{\mathrm{th}}}{2}.
\end{equation}
The likelihood ratio for these hypotheses is
\begin{equation}
\Lambda(x) = \frac{p(x \mid p_{\text{bin}} = p_1)}{p(x \mid p_{\text{bin}} = p_0)} 
= \frac{\binom{n}{x} p_1^x (1 - p_1)^{n-x}}{\binom{n}{x} p_0^x (1 - p_0)^{n-x}} 
= \qty(\frac{p_1 (1 - p_0)}{p_0 (1 - p_1)})^x \qty(\frac{1 - p_1}{1 - p_0})^n.
\end{equation}
Since $p_1 > p_0$, it follows that
\begin{equation}
\frac{p_1 (1 - p_0)}{p_0 (1 - p_1)} > \frac{p_0 (1 - p_1)}{p_0 (1 - p_1)} = 1.
\end{equation}
Thus, $\Lambda(x)$ is a monotonically increasing function of the statistic $T(x) = x$, i.e., the MLR in $T(x) = x$. 
By Lemma~\ref{lemma:karlin-rubin}, the test function defined in Eq.~(\ref{eq:test of thm:UMP pauli string}) is the UMP test for all $\ev{O}$ in both cases (i) and (ii).
At a given significance level $\alpha$, $c$ is the smallest non-negative integer satisfying 
\begin{equation}
\Pr(x > c \mid \ev{O} = \ev{O}_{\mathrm{th}}) = \Pr(x > c \mid p_{\text{bin}} = \frac{1 + \ev{O}_{\mathrm{th}}}{2}) 
= \sum_{x > c} \binom{n}{x} \qty(\frac{1 + \ev{O}_{\mathrm{th}}}{2})^x \qty(\frac{1 - \ev{O}_{\mathrm{th}}}{2})^{n-x} \leq \alpha,
\end{equation}
and $\gamma \in [0, 1]$ is chosen to satisfy
\begin{equation}
\begin{split}
&\Pr(x > c \mid \ev{O} = \ev{O}_{\mathrm{th}}) + \gamma \Pr(x = c \mid \ev{O} = \ev{O}_{\mathrm{th}}) \\
=& \Pr(x > c \mid p_{\text{bin}} = \frac{1 + \ev{O}_{\mathrm{th}}}{2}) + \gamma \Pr(x = c \mid p_{\text{bin}} = \frac{1 + \ev{O}_{\mathrm{th}}}{2}) \\
=& \sum_{x > c} \binom{n}{x} \qty(\frac{1 + \ev{O}_{\mathrm{th}}}{2})^x \qty(\frac{1 - \ev{O}_{\mathrm{th}}}{2})^{n-x} + \gamma \binom{n}{c} \qty(\frac{1 + \ev{O}_{\mathrm{th}}}{2})^c \qty(\frac{1 - \ev{O}_{\mathrm{th}}}{2})^{n-c} \\
=& \alpha.
\end{split}
\end{equation}
\end{proof}

By performing the projective measurement described by the order parameter $O_{\text{SPT}}$ $n$ times, we test whether the expectation value $\ev{O_{\text{SPT}}}$ is zero or not. 
The test data $\{ \rho_{\text{test}}^{(i)} \}_{i=1}^{N_{\text{test}}}$ that we used all have expectation values of $O_{\text{SPT}}$ greater than or equal to zero. 
Thus, the hypothesis corresponding to the trivial phase (labeled $y^{(i)} = 0$) is $H_0: \text{Tr}(\rho_{\text{test}}^{(i)} O_{\text{SPT}}) = 0$, and the hypothesis corresponding to the SPT phase (labeled $y^{(i)} = 1$) is $H_1: \text{Tr}(\rho_{\text{test}}^{(i)} O_{\text{SPT}}) > 0$.
The test described in case (i) of Thm.~\ref{thm:UMP pauli string} with $\ev{O}_{\text{th}} = 0$ is therefore UMP for these hypotheses. 
We refer to this test as the classical Neyman-Pearson test for the SPT phase.

Here, the Type-\Romannumeral{1} error probability is given by
\begin{equation}
\alpha_n^{(i)} = \sum_{x > c} \binom{n}{x} \qty(\frac{1 + \ev{O}^{(i)}}{2})^x 
\qty(\frac{1 - \ev{O}^{(i)}}{2})^{n-x} 
+ \gamma \binom{n}{c} \qty(\frac{1 + \ev{O}^{(i)}}{2})^c 
\qty(\frac{1 - \ev{O}^{(i)}}{2})^{n-c},
\end{equation}
where $\ev{O}^{(i)} = \text{Tr}(\rho_{\text{test}}^{(i)} O_{\text{SPT}})$ for test data labeled $y^{(i)} = 0$.
The Type-\Romannumeral{2} error probability is given by
\begin{equation}
\beta_n^{(i)} = \sum_{x < c} \binom{n}{x} \qty(\frac{1 + \ev{O}^{(i)}}{2})^x 
\qty(\frac{1 - \ev{O}^{(i)}}{2})^{n-x} 
+ (1 - \gamma) \binom{n}{c} \qty(\frac{1 + \ev{O}^{(i)}}{2})^c 
\qty(\frac{1 - \ev{O}^{(i)}}{2})^{n-c},
\end{equation}
where $\ev{O}^{(i)} = \text{Tr}(\rho_{\text{test}}^{(i)} O_{\text{SPT}})$ for test data labeled $y^{(i)} = 1$.
The error probabilities $\alpha_n$ and $\beta_n$ computed in the main text are the averages of $\alpha_n^{(i)}$ and $\beta_n^{(i)}$ over the test data for each phase, respectively.

\paragraph*{Bayesian test using the order parameter of the FM phase}

In the context of Trivial vs. FM, the classification between the trivial phase and the FM phase is performed by testing whether the expectation value of the order parameter $O_{\text{FM}}$, defined as a linear combination of local observables, is zero or not. 
While the optimal classical post-processing for the SPT order parameter, QCNN, and Exact QCNN after measurements is determined based on Thm.~\ref{thm:UMP pauli string}, the same does not hold for $O_{\text{FM}}$. 
This is because, unlike other cases where the measurement results follow a Bernoulli distribution, the measurement outcomes associated with $O_{\text{FM}}$ follow a multinomial distribution. 
For multinomial distributions, there is no general result providing a UMP test analogous to Lemma~\ref{lemma:karlin-rubin}.
Therefore, since knowing the order parameter implies possessing prior knowledge about quantum phases, we employ a Bayesian test to make effective use of this knowledge. 

The eigenvalue decomposition of the order parameter $O_{\text{FM}}$ is given by
\begin{equation}
O_{\text{FM}} = \frac{1}{L} \sum_{i=1}^{L} Z_i = \sum_{m=0}^{L} \lambda_m \Pi_m,
\end{equation}
where $\lambda_m = 1 - 2\frac{m}{L}$ are the eigenvalues, and $\Pi_m = \sum_{|\lambda|=m} \ketbra{\lambda}$ represents the projection onto the subspace spanned by computational basis states with Hamming weight $|\lambda| = m$. 
Performing the projective measurement described by $O_{\text{FM}}$ $n$ times, we denote $X_m$ as the number of times $\Pi_m$ is observed and define the random variable $\mathbf{X} = (X_0, ..., X_L)$.
To formalize this setting, we introduce the multinomial and Dirichlet distributions. 
The multinomial distribution $\text{Mul}(n, \mathbf{p})$ describes the probability of obtaining category counts $\mathbf{X} = (X_0, ..., X_L)$ given event probabilities $\mathbf{p} = (p_0, ..., p_L)$, where $p_m \geq 0$ and $\sum_{m=0}^{L} p_m = 1$. 
The probability mass function is
\begin{equation}
p_{\text{Mul}}(\mathbf{x} \mid n, \mathbf{p}) = \frac{n!}{x_0! \dots x_L!} \prod_{m=0}^{L} p_m^{x_m}, 
\end{equation}
where $\sum_{m=0}^{L} x_m = n$.
The Dirichlet distribution $\text{Dir}(\mathbf{\alpha})$ is defined over the $L$-dimensional simplex $\mathbf{P} = (P_0, ..., P_L)$, where $P_m \geq 0$ and $\sum_{m=0}^{L} P_m = 1$. 
It is parameterized by pseudo-counts $\mathbf{\alpha} = (\alpha_0, ..., \alpha_L)$, with the probability density function
\begin{equation}
p_{\text{Dir}}(\mathbf{p} \mid \mathbf{\alpha}) = \frac{1}{B(\mathbf{\alpha})} \prod_{m=0}^{L} p_m^{\alpha_m - 1},
\end{equation}
where $B(\mathbf{\alpha})$ is the multivariate Beta function given by
\begin{equation}
B(\mathbf{\alpha}) = \frac{\prod_{m=0}^{L} \Gamma(\alpha_m)}{\Gamma\qty(\sum_{m=0}^{L} \alpha_m)},
\end{equation}
with $\Gamma(\cdot)$ denoting the Gamma function.

The hypothesis corresponding to the trivial phase (labeled $y^{(i)} = 0$) is $H_0: \text{Tr}(\rho_{\text{test}}^{(i)} O_{\text{FM}}) = 0$, and the true distribution of $\mathbf{X}$ follows $\text{Mul}(n, \mathbf{p}=(\text{Tr}(\rho_{\text{test}}^{(i)} \Pi_0), ..., \text{Tr}(\rho_{\text{test}}^{(i)} \Pi_L)), H_0) = \text{Mul}_0^{(i)}$. 
As a prior for the multinomial probabilities $\mathbf{P}$, we choose a Dirichlet distribution $\text{Dir}(\mathbf{\alpha} = \mathbf{\alpha}_0, H_0)$, where
\begin{equation}
\mathbf{\alpha}_0 = \frac{1}{2^L} \qty(\binom{L}{0}, ..., \binom{L}{L}).
\end{equation}
This choice is motivated by the fact that the trivial phase in the one-dimensional cluster-Ising model is dominated by $\sum_{i=1}^{L} X_i$, whose ground states are uniform superpositions of computational basis states. 
Since each $\Pi_m$ represents a sum of projections onto $\binom{L}{m}$ individual computational basis states, we choose the Dirichlet prior so that $\alpha_m$ is proportional to $\binom{L}{m}$.
The marginal likelihood under $H_0$ is then computed as
\begin{equation}
\begin{split}
p(\mathbf{x} \mid n, H_0) &= \int p_{\text{Mul}}(\mathbf{x} \mid n, \mathbf{p}) p_{\text{Dir}}(\mathbf{p} \mid \mathbf{\alpha}
= \mathbf{\alpha}_0, H_0) d\mathbf{p}\\
&= \int \frac{n!}{x_0! \dots x_L!} \prod_{m=0}^{L} p_m^{x_m} \frac{1}{B(\mathbf{\alpha}_0)} \prod_{m=0}^{L} p_m^{\alpha_m - 1} d\mathbf{p}\\
&= \frac{n!}{x_0! \dots x_L!} \frac{1}{B(\mathbf{\alpha}_0)} \int \prod_{m=0}^{L} p_m^{\alpha_m + x_m - 1} d\mathbf{p}\\
&= \frac{n!}{x_0! \dots x_L!} \frac{B(\mathbf{\alpha}_0 + \mathbf{x})}{B(\mathbf{\alpha}_0)}.
\end{split}
\end{equation}
On the other hand, the hypothesis corresponding to the FM phase (labeled $y^{(i)} = 1$) is $H_1: \text{Tr}(\rho_{\text{test}}^{(i)} O_{\text{FM}}) \neq 0$, and the true distribution of $\mathbf{X}$ follows $\text{Mul}(n, \mathbf{p}=(\text{Tr}(\rho_{\text{test}}^{(i)} \Pi_0), ..., \text{Tr}(\rho_{\text{test}}^{(i)} \Pi_L)), H_1) = \text{Mul}_1^{(i)}$. 
As a prior, we choose $\text{Dir}(\mathbf{\alpha} = \mathbf{\alpha}_1, H_1)$, where
\begin{equation}
\mathbf{\alpha}_1 = \frac{2}{(L+1)(L+2)}(1, ..., L+1) \quad \text{or} \quad \frac{2}{(L+1)(L+2)}(L+1, ..., 1).
\end{equation}
This choice is motivated by the fact that the FM phase in the one-dimensional cluster-Ising model is dominated by $-\sum_{i=1}^{L} Z_i Z_{i+1}$, whose ground states are superpositions of $\ket{0}^{\otimes L}$ and $\ket{1}^{\otimes L}$. 
The observed $\Pi_m$ values are thus more likely to be biased toward either small or large $m$.
The marginal likelihood under $H_1$ is 
similarly given by
\begin{equation}
p(\mathbf{x} \mid n, H_1) = \frac{n!}{x_0! \dots x_L!} \frac{B(\mathbf{\alpha}_1 + \mathbf{x})}{B(\mathbf{\alpha}_1)}.
\end{equation}

In Bayesian hypothesis testing, the Bayes factor is defined as the ratio of marginal likelihoods,
\begin{equation}
\text{BF}_{10} = \frac{p(\mathbf{x} \mid n, H_1)}{p(\mathbf{x} \mid n, H_0)} = \frac{B(\mathbf{\alpha}_1 + \mathbf{x})B(\mathbf{\alpha}_0)}{B(\mathbf{\alpha}_0 + \mathbf{x})B(\mathbf{\alpha}_1)},
\end{equation}
and the test is performed by comparing $\text{BF}_{10}$ to a threshold $c$. 
Since the prior distribution under $H_1$ is not uniquely determined, we adopt the Bayes factor corresponding to the distribution that deviates most from the threshold.
The Type-\Romannumeral{1} and Type-\Romannumeral{2} error probabilities are then given by
\begin{equation}
\alpha_n^{(i)} = \Pr(\text{BF}_{10} \geq c \mid \mathbf{X} \sim \text{Mul}_0^{(i)}, H_0), \quad
\beta_n^{(i)} = \Pr(\text{BF}_{10} < c \mid \mathbf{X} \sim \text{Mul}_1^{(i)}, H_1).
\end{equation}
The error probabilities $\alpha_n$ and $\beta_n$ computed in the main text are the averages of $\alpha_n^{(i)}$ and $\beta_n^{(i)}$ over the test data for each phase, respectively.

\subsection{QCNN and Exact QCNN}
\label{app:QCNN and cNeyman}
\paragraph*{Exact QCNN for the FM phase}

Exact QCNNs are constructed based on the Multiscale Entanglement Renormalization Ansatz (MERA)~\cite{Vidal2008} and the Multiscale String Operator (MSO)~\cite{Cong2019,Lake2025}. 
The MSO is defined as a sum of products of exponentially many string order parameters that act on different locations. 
For example, for the SPT phase, the MSO is expressed as
\begin{equation}\label{eq:MSO}
\text{MSO} = \sum_{ab} C_{ab}^{(1)} \mathcal{S}_{ab} + \sum_{a_1b_1a_2b_2} C_{a_1b_1a_2b_2}^{(2)} \mathcal{S}_{a_1b_1} \mathcal{S}_{a_2b_2} + \cdots,
\end{equation}
where the string order parameter $O_{\text{SPT}}$ in the main text is redefined as $\mathcal{S}_{ab} = Z_a X_{a+1} X_{a+3} \dots X_{b-3} X_{b-1} Z_b$. 
The observable in the Heisenberg picture for the Exact QCNN proposed in Ref.~\cite{Cong2019} (Fig.~\ref{fig:Exact QCNN circuits}(b)) matches this expression.

We consider an Exact QCNN for the FM phase characterized by the order parameter that is the linear combination of local observables, as $O_{\text{FM}}$ in the main text. 
Since this order parameter is not a string, we instead consider the long-range order parameter $\mathcal{S}_{ab} = Z_a Z_b$ and construct a circuit such that the observable in the Heisenberg picture matches Eq.~(\ref{eq:MSO}).
The Exact QCNN circuit for the FM phase is shown in Fig.~\ref{fig:Exact QCNN for FM}(a). 
The unitary operation combining the convolution and pooling layers at a certain depth $d$, $U_{\text{CP}}^{(d)}$, is defined as shown in Fig.~\ref{fig:Exact QCNN for FM}(b). 
Using the relations for Pauli propagation illustrated in Fig.~\ref{fig:Exact QCNN for FM}(c), the observable in the Heisenberg picture for this Exact QCNN becomes
\begin{equation}
\begin{split}
&(U_{\text{CP}}^{(depth)}\cdots U_{\text{CP}}^{(1)})^\dagger (I\otimes\cdots\otimes I\otimes Z\otimes I\otimes\cdots\otimes I\otimes Z) (U_{\text{CP}}^{(depth)}\cdots U_{\text{CP}}^{(1)}) \\
=& \sum_{ab} C_{ab}^{(1)}Z_aZ_b + \sum_{a_1b_1a_2b_2} C_{a_1b_1a_2b_2}^{(2)}Z_{a_1}Z_{b_1}Z_{a_2}Z_{b_2} + \cdots \\
=& \sum_{ab} C_{ab}^{(1)}\mathcal{S}_{ab} + \sum_{a_1b_1a_2b_2} C_{a_1b_1a_2b_2}^{(2)}\mathcal{S}_{a_1b_1}\mathcal{S}_{a_2b_2} + \cdots, 
\end{split}
\end{equation}
which matches the MSO in Eq.~(\ref{eq:MSO}).

In the main text (the subsection on error probabilities), it is evident that under a few copies of test data, the Exact QCNN for the FM phase performs worse than the order parameter. 
This suggests that Exact QCNNs may not be advantageous for quantum phases characterized by order parameters expressible as linear combinations of local observables.

\begin{figure}[t]
    \centering
    \includegraphics[width=0.6\columnwidth]{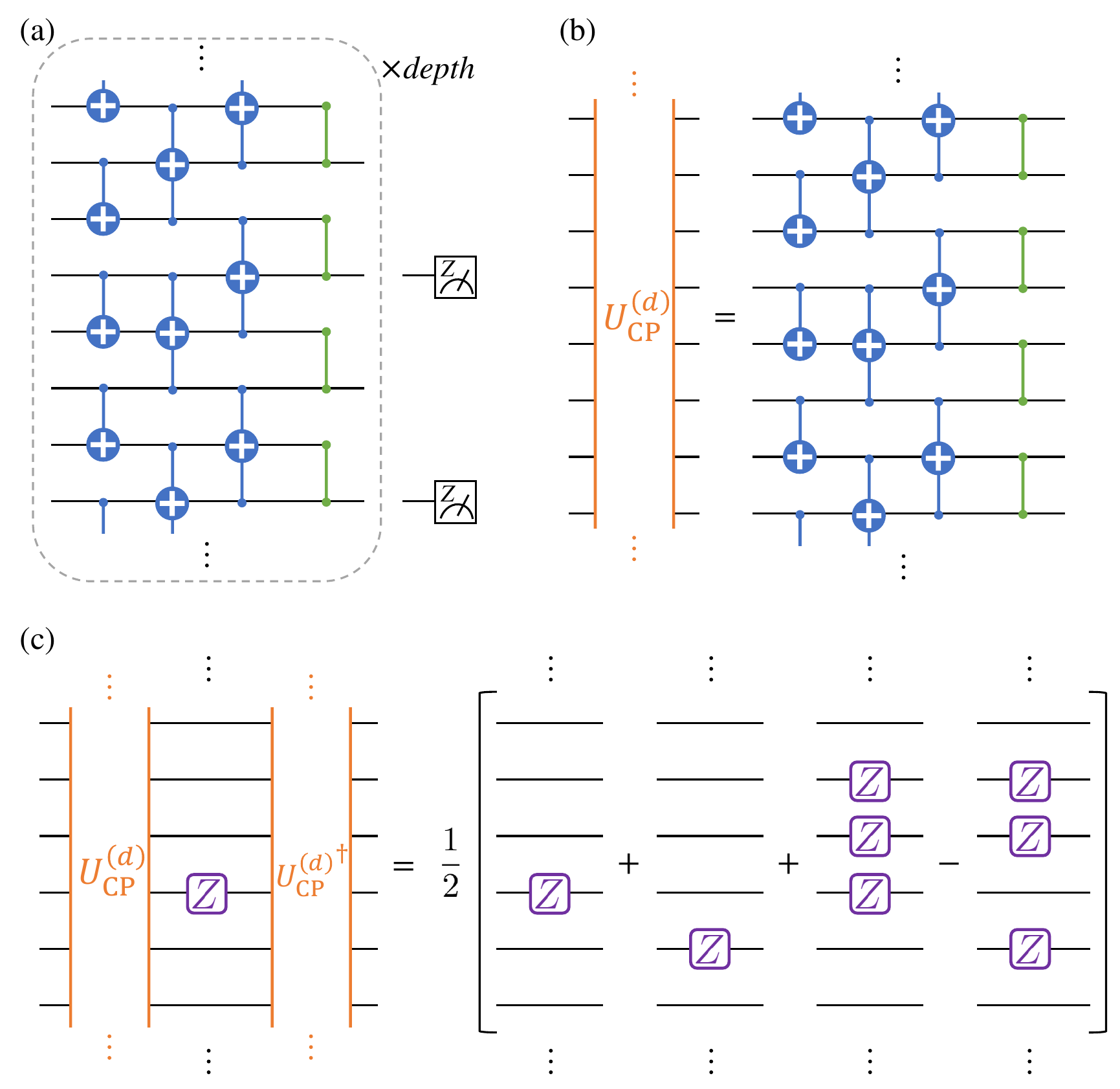}
    \caption{Exact QCNN for the FM phase. 
    Panel (a) shows the full Exact QCNN circuit, identical to Fig.~\ref{fig:Exact QCNN circuits}(a), where the convolution and pooling layers at a certain depth $d$ are denoted as $U_{\text{CP}}^{(d)}$, as shown in panel (b). 
    Panel (c) illustrates the equality ${U_{\text{CP}}^{(d)}}^\dagger Z_i U_{\text{CP}}^{(d)} = \frac{1}{2}(Z_i + Z_{i+1} + Z_{i-2}Z_{i-1}Z_i - Z_{i-2}Z_{i-1}Z_{i+1})$ obtained using Pauli propagation. 
    The blue two-qubit gates represent controlled--$X$ gates, the blue three-qubit gates represent Toffoli gates, the green two-qubit gates represent controlled--$Z$ gates, and the $ZZ$ measurement indicates a projective measurement in the computational basis.}
    \label{fig:Exact QCNN for FM}
\end{figure}

\paragraph*{Classical Neyman-Pearson test using QCNNs and Exact QCNNs}

The QCNN and the Exact QCNN classify quantum phases based on the expectation values of Pauli operators or Pauli strings with respect to the output quantum states (e.g., $ZXZ$ Pauli string: the Exact QCNN for the Trivial vs. SPT case). 
By performing the projective measurement described by a Pauli operator or Pauli string $n$ times, we test whether the output expectation value $f(\rho_{\text{test}}^{(i)})$ for input test data $\rho_{\text{test}}^{(i)}$ exceeds $0.5$ or not. 
The hypothesis corresponding to the trivial phase (labeled $y^{(i)} = 0$) is $H_0: f(\rho_{\text{test}}^{(i)}) \leq 0.5$, and the hypothesis corresponding to the non-trivial phase (labeled $y^{(i)} = 1$) is $H_1: f(\rho_{\text{test}}^{(i)}) > 0.5$.
The test described in case (ii) of Thm.~\ref{thm:UMP pauli string} with $\ev{O}_{\text{th}} = 0.5$ is therefore UMP for these hypotheses. 
We refer to this test as the classical Neyman-Pearson test for QCNNs and Exact QCNNs.

The Type-\Romannumeral{1} error probability for test data labeled $y^{(i)} = 0$ is given by
\begin{equation}
\alpha_n^{(i)} = \sum_{x > c} \binom{n}{x} \qty(\frac{1 + f(\rho_{\text{test}}^{(i)})}{2})^x 
\qty(\frac{1 - f(\rho_{\text{test}}^{(i)})}{2})^{n-x} 
+ \gamma \binom{n}{c} \qty(\frac{1 + f(\rho_{\text{test}}^{(i)})}{2})^c 
\qty(\frac{1 - f(\rho_{\text{test}}^{(i)})}{2})^{n-c},
\end{equation}
and the Type-\Romannumeral{2} error probability for test data labeled $y^{(i)} = 1$ is given by
\begin{equation}
\beta_n^{(i)} = \sum_{x < c} \binom{n}{x} \qty(\frac{1 + f(\rho_{\text{test}}^{(i)})}{2})^x 
\qty(\frac{1 - f(\rho_{\text{test}}^{(i)})}{2})^{n-x} 
+ (1 - \gamma) \binom{n}{c} \qty(\frac{1 + f(\rho_{\text{test}}^{(i)})}{2})^c 
\qty(\frac{1 - f(\rho_{\text{test}}^{(i)})}{2})^{n-c}.
\end{equation}
The error probabilities $\alpha_n$ and $\beta_n$ computed in the numerical results are the averages of $\alpha_n^{(i)}$ and $\beta_n^{(i)}$ over the test data for each phase, respectively.

\end{document}